\documentclass[11pt]{article}
\usepackage[margin=0.75in]{geometry}
\usepackage{amsmath}
\usepackage{graphicx}
\usepackage{comment}
\usepackage{enumerate}

\usepackage{url} 

\newcommand{\blind}{1}

\usepackage{amsmath,amsthm, amssymb, bbm}
\usepackage{blindtext}
\usepackage{enumerate}
\usepackage{listings}
\usepackage{color, colortbl}
\usepackage{xcolor}
\usepackage[nottoc]{tocbibind}
\usepackage[utf8]{inputenc}
\usepackage{caption}
\usepackage{tabularx}
\usepackage{array}
\usepackage{rotating}
\usepackage{wrapfig}
\usepackage{graphicx}
\usepackage{subcaption}
\definecolor {processblue}{cmyk}{0.96,0,0,0}
\usepackage{mathrsfs}
\usepackage{hyperref}
\hypersetup{
	colorlinks   = true,
	linkcolor = red,
	citecolor    = blue
}
\usepackage[square]{natbib} 
\usepackage{fancyvrb}
\usepackage{enumitem}
\definecolor{LightCyan}{rgb}{0.88,1,1}
\usepackage{parskip}
\usepackage{float}
\usepackage{titling}
\usepackage{varwidth}
\usepackage{multirow}
\usepackage{breqn}
\usepackage{bm}
\usepackage{titlesec}
\titleformat{\subsection}    
{\bfseries\itshape}{\thesubsection}{1em}{}
\usepackage{xr}
\usepackage{setspace}

\newtheorem{theorem}{Theorem}
\newtheorem{lemma}{Lemma}
\newtheorem{proposition}{Proposition}

\newtheorem{corollary}{Corollary}
\theoremstyle{remark}
\newtheorem{definition}{Definition}
\newtheorem{example}{Example}
\newtheorem{assumption}{Assumption}
\newtheorem{remark}{Remark}
\usepackage{amssymb,mathtools}

\makeatletter
\newcommand*{\addFileDependency}[1]{
	\typeout{(#1)}
	\@addtofilelist{#1}
	\IfFileExists{#1}{}{\typeout{No file #1.}}
}
\makeatother

\newcommand{\var}{\text{var}}

\newcommand{\Pb}{\mathbb{P}}

\newcommand{\Pn}{\mathbb{P}_n}
\newcommand{\Un}{\mathbb{U}_n}

\newcommand{\E}{\mathbb{E}}
\newcommand{\R}{\mathbb{R}}

\def\expit{\text{expit}}

\newcommand{\one}{\mathbbm{1}}

\begin{document}

	\def\spacingset#1{\renewcommand{\baselinestretch}%
		{#1}\small\normalsize} \spacingset{1}

	
	\if1\blind
	{
		\title{\bf Minimax optimal subgroup identification}
		\author{ Matteo Bonvini\thanks{Department of Statistics \& Data Science, Carnegie Mellon University, 5000 Forbes Avenue, Pittsburgh, PA 15213. Email: matteobonvini@gmail.com.} \and Edward H. Kennedy\thanks{Associate Professor, Department of Statistics \& Data Science, Carnegie Mellon University, 5000 Forbes Avenue, Pittsburgh, PA 15213. Email: edward@stat.cmu.edu.} \and Luke J. Keele\thanks{Research Associate Professor, Department of Surgery, Perelman School of Medicine, University of Pennsylvania, 3400 Spruce Street, 4 Silverstein, Philadelphia, PA 19104.	Email: luke.keele@uphs.upenn.edu}}
		\date{ \today \\ \medskip \textit{Preliminary draft. Comments welcome.}}
		\maketitle
	} \fi
	
	\if0\blind
	{
		\bigskip
		\bigskip
		\bigskip
		\begin{center}
			{\bf Title}
		\end{center}
		\medskip
	} \fi
	
	\bigskip
	\begin{abstract}
		Quantifying treatment effect heterogeneity is a crucial task in many areas of causal inference, e.g. optimal treatment allocation and estimation of subgroup effects. We study the problem of estimating the level sets of the conditional average treatment effect (CATE), identified under the no-unmeasured-confounders assumption. Given a user-specified threshold, the goal is to estimate the set of all units for whom the treatment effect exceeds that threshold. For example, if the cutoff is zero, the estimand is the set of all units who would benefit from receiving treatment. Assigning treatment just to this set represents the optimal treatment rule that maximises the mean population outcome. Similarly, cutoffs greater than zero represent optimal rules under resource constraints. Larger cutoffs can also be used for anomaly detection, i.e., finding which subjects are most affected by treatments. Being able to accurately estimate CATE level sets is therefore of great practical relevance. The level set estimator that we study follows the plug-in principle and consists of simply thresholding a good estimator of the CATE. While many CATE estimators have been recently proposed and analysed, how their properties relate to those of the corresponding level set estimators remains unclear. Our first goal is thus to fill this gap by deriving the asymptotic properties of level set estimators depending on which estimator of the CATE is used. Next, we identify a minimax optimal estimator in a model where the CATE, the propensity score and the outcome model are H\"{o}lder-smooth of varying orders. We consider data generating processes that satisfy a margin condition governing the probability of observing units for whom the CATE is close to the threshold. We investigate the performance of the estimators in simulations and illustrate our methods on a dataset used to study the effects on mortality of laparoscopic vs open surgery in the treatment of various conditions of the colon. 
	\end{abstract}

\section{Introduction}

Much empirical research focuses on estimating causal effects. One commonly estimated causal effect is the average treatment effect (ATE), which is the difference in average outcome if everyone in the population, versus no one, receives treatment. By definition, the ATE is an aggregate measure of treatment efficacy that does not capture any effect heterogeneity. An alternative measure of treatment effect is the conditional average treatment effect (CATE), which is the ATE restricted to a subpopulation of interest. The subpopulation is typically defined by the values of some \textit{a priori} selected variables known as \textit{effect modifiers}. One natural extension of the CATE is to estimate the set of units with treatment effects larger (or smaller) than some user-specified threshold. For example, when the threshold is zero, assigning treatment to only those units with a positive treatment effect is the optimal rule maximizing the mean outcome in the population (see, e.g, \cite{robins2004optimal}, \cite{hirano2009asymptotics}, \cite{chakraborty2013statistical}, and \cite{luedtke2016statistical}). 

To consider this problem, informally, we define $Y$ as the outcome, $A$ as an indicator for treatment, and $X$ as measured confounders and simultenously effect modifiers. Using these terms, the CATE $\tau(x)$ is equal to $\tau(x) = \E(Y \mid A = 1, X = x) - \E(Y \mid A = 0, X = x)$, and the ATE is $\E\{\tau(X)\}$. Our target of inference, the upper level set of the CATE at $\theta$, is
\begin{align*}
	\Gamma(\theta) = \{x \in \mathcal{X}: \tau(x) > \theta\}
\end{align*}
We assume the level $\theta \in \R$ to be user-specified. For some estimator $\widehat\tau(x)$ of $\tau(x)$, we estimate the level set, $\Gamma(\theta)$, with 
\begin{align*}
	\widehat\Gamma(\theta) = \{x \in \mathcal{X}: \widehat\tau(x) > \theta\}.
\end{align*}
This estimator defines the set of study units with estimated CATEs that are greater than $\theta$. Clearly, an estimator for $\Gamma(\theta)$ depends on an estimator for $\tau(x)$, and the performance of $\widehat\Gamma(\theta)$ will be affected by how well $\tau(x)$ can be estimated. Yet, we will show that estimating $\Gamma(\theta)$ can be an easier statistical problem than estimating $\tau(x)$ itself on its support. Intuitively, one needs to be able to estimate $\tau(x)$ accurately only in regions of the covariates' space where $\tau(x)$ is close to $\theta$. Further, if $\tau(X)$ has a bounded density and a particular loss function is used, we will show that the convergence rate of $\widehat\Gamma(\theta)$ to $\Gamma(\theta)$ will generally be faster than that of $\widehat\tau(x)$ to $\tau(x)$. 

Recent work has developed a number of proposals for CATE estimation with an emphasis on using nonparametric methods borrowed from the machine learning (ML) literature \citep{athey2016recursive, foster2019orthogonal, semenova2021debiased, nie2021quasi, kennedy2020optimal, kennedy2022minimax, wager2018estimation, kunzel2019metalearners, hahn2020bayesian, imai2013estimating, shalit2017estimating}. In our work, we focus on a class of nonparametric estimators for the CATE that are embedded in a meta-learner framework that separates estimation of the CATE into a multi-step regression procedure. In the first step, a set of nuisance functions is estimated using flexible machine learning models. Then, in the second-stage, an estimate of $\tau(x)$ is computed using the nuisance function estimates as inputs. More specifically, we focus on two recently proposed estimators of $\tau(x)$: the DR-Learner analyzed in \cite{kennedy2020optimal} and the Lp-R-Learner proposed in \cite{kennedy2022minimax}. The first one is a general estimation procedure based on a two-stage regression that can be computed using off-the-shelf software. The second is a more complicated estimator, which has been shown to be minimax optimal for an important set of models. 

We merge this work on flexible estimation of CATEs with the extensive literature on nonparametric estimation of (upper) level sets.
See, for examples, \citet{qiao2019nonparametric}, \citet{mammen2013confidence}, \citet{chen2017density}, \citet{rigollet2009optimal}, \citet{willett2007minimax} and references therein. The main difference between our work and this research is that in our context the level set is defined by the difference of two regressions, the optimal estimation of which can be considerably more involved than that of either regression. Within this literature, our work is closest to \cite{rigollet2009optimal}, and we use their general framework to analyze the performance of our estimators. 

Other streams of research closely related to our work are policy learning \citep{hirano2009asymptotics, athey2021policy, ben2022policy} and contextual bandits \citep{gur2022smoothness}. In the policy learning literature, it is typically assumed that the best policy belongs to some well-behaved and interpretable class of decision rules. This is different from the route we take in this work; instead of restricting the complexity of the level set class, we restrict the complexity of the CATE function in nonparametric models. In addition, while one of the core goals of the literature on contextual bandits is to identify regions of the covariates' space where the treatment effect is positive, this is usually done in settings where  the probability of taking a given action or receiving treatment, i.e., the propensity score, is under the experimenter's control and known. Instead, we consider observational studies where the propensity score is unknown. Finally, \cite{reeve2021optimal} has also considered a similar problem to the one discussed in this paper, but they require that the propensity score is known and the estimator appears to be more complicated.

\subsection{Our contribution}

The level set estimator that we study follows the plug-in principle and consists of simply thresholding an estimator of the CATE. To the best of our knowledge, how the properties of a CATE learner relate to those of the corresponding level set estimator has not been investigated in the literature yet. As such, our first goal is to derive the asymptotic properties of level set estimators depending on which estimator of the CATE is used.  We calculate the risk for estimating $\Gamma(\theta)$ by thresholding a general estimator of the CATE required to satisfy a particular exponential inequality. Then, we specialize the results when the CATE is estimated with the DR-Learner or the Lp-R-Learner. Further, we show that if the Lp-R-Learner is used, the risk achieved is minimax optimal, under certain conditions. The optimality of thresholding the Lp-R-Learner for estimating CATE level sets had yet to be established. As an intermediate step for obtaining our main results, we derive exponential inequalities for CATE estimators based on linear smoothing, which might be of independent interest.

We establish the minimax optimal rate for estimating $\Gamma(\theta)$ in H\"{o}lder smoothness models where $\tau(x)$ and the nuisance functions have potentially different smoothness levels. \cite{kennedy2022minimax} have recently shown that, from a minimax optimality point of view, the parameter $\tau(x)$ shares features of a functional with nuisance components \citep{robins2009semiparametric, robins2017minimax} and a standard nonparametric regression \citep{tsybakov2004introduction}. Building upon their work and \cite{rigollet2009optimal}, we show that $\Gamma(\theta)$ behaves as a hybrid parameter not only exhibiting features similar to those of $\tau(x)$, but also those of a Bayes classifier. Effectively, we connect the problem of estimating CATE level sets to the domains of classification, nonparametric regression and functional estimation. We also briefly discuss the construction of confidence sets for $\Gamma(\theta)$ based on the distribution of $\sup_{x\in \mathcal{X}}|\widehat\tau(x) - \tau(x)|$. Finally, we illustrate our methods in simulations and with a real dataset used to study the effect of laparoscopic surgery for partial colectomy on mortality and complications. 
\section{Notation}
We assume that $X$ has at least one continuous component and denote the marginal CDF of X by $F(x)$, with corresponding density $f(x)$ with respect to the Lebesgue measure, which we assume to be uniformly bounded. We also let $d$ denote the dimension of $X$ and let $\mathcal{X}$ be the set of all $x \in \R^d$ such that $f(x) > 0$. 

We define the nuisance functions:
\begin{align*}
	& \pi(X) = \Pb(A = 1 \mid X), \ \text{the probability of receiving treatment or propensity score}, \\
	& \mu(X) = \E(Y \mid X), \ \text{the outcome model}, \\
	& \mu_a(X) = \E(Y \mid A = a, X),  \ \text{the outcome model for units treated with } A = a, \\
	& \text{and } \tau(X) = \mu_1(X) - \mu_0(X), \ \text{the CATE identified under no-unmeasured-confounding}.
\end{align*}
Unless we need to keep track of constants, we will adopt the notation $a \lesssim b$ to mean that there exists a constant $C$ such that $a \leq C b$. We assume all the nuisance functions are uniformly bounded and $\pi(x)$ is also bounded away from 0 and 1. To keep the notation as light as possible, we will write $\Gamma$ to mean $\Gamma(\theta)$. 

Let $s = (s_1, \ldots, s_d) \in \mathbb{N}^d$, $|s|= \sum_{i = 1}^d s_i$, $s! = s_1! \cdots s_d!$ and $D^s = \frac{\partial^{s_1 + \ldots s_d}}{\partial x_1^{s_1} \cdots \partial x_d^{s_d}}$ be the differential operator. For $\beta > 0$, let $\lfloor \beta \rfloor$ denote the largest integer strictly less than $\beta$. Given $x \in \R^d$ and $f$ a $\lfloor \beta \rfloor$-times continuously differentiable function, let 
$$f_x(u) = \sum_{|s| \leq \lfloor \beta \rfloor} \frac{(u - x)^s}{s!} D^s f(x)$$
denote its Taylor polynomial approximation of order $\lfloor \beta \rfloor$ at $u = x$.
\begin{definition}[locally H\"{o}lder-$\beta$ function]\label{def:holder}
	A function $f$ is ``$\beta$-smooth locally around a point $x_0 \in \mathcal{X}_0$" if it is $\lfloor \beta \rfloor$-times continuosly differentiable at $x_0$ and there exists a constant $L$ such that
	\begin{align*}
		\left|f(x) - f_{x_0} (x)\right| \leq L \|x - x_0\|^{\beta} \text{ for all } x \in B(x_0, r), r > 0. 
	\end{align*}
\end{definition}
There are a few ways to measure the performance of $\widehat\Gamma(\theta)$, two of which are
\begin{itemize}
	\item $d_\Delta(\widehat\Gamma, \Gamma) = \int_{\widehat\Gamma \Delta \Gamma} f(x) dx$, for $\widehat\Gamma \Delta \Gamma = (\widehat\Gamma^c \cap \Gamma) \cup (\widehat\Gamma \cap \Gamma^c)$ (set difference);
	\item $d_H(\widehat\Gamma, \Gamma) = \int_{\widehat\Gamma \Delta \Gamma} |\tau(x) - \theta| f(x) dx$ (penalized set difference).
\end{itemize}
The first one is simply the $\Pb_X$-measure of the set difference between $\widehat\Gamma$ and $\Gamma$. The second one is the $\Pb_X$-measure of the set difference simply with a smaller penalty assigned to errors made by including / excluding values of $X$ for which the CATE is close to $\theta$. In particular, whether or not $x$ such that $\tau(x) = \theta$ is included or excluded from the set $\widehat\Gamma$ has no impact on the error measured by $d_H(\widehat\Gamma, \Gamma)$. If $\theta = 0$, this means that, according to this metric, it does not matter whether we assign treatment to units with zero treatment effect. 
We will focus on $d_H(\widehat\Gamma, \Gamma)$ and analyze the risk
\begin{align} \label{eq:risk}
	\E\left\{d_H(\widehat\Gamma, \Gamma)\right\} = \E\left\{ \int_{\widehat\Gamma \Delta \Gamma} |\tau(x) - \theta| f(x) dx \right\},
\end{align}
which we represent in Figure \ref{fig:loss} for the case $d = 1$ and $X \sim \text{Unif}(0, 1)$. 
\begin{figure}[!h]
	\centering
	\includegraphics[scale=0.5]{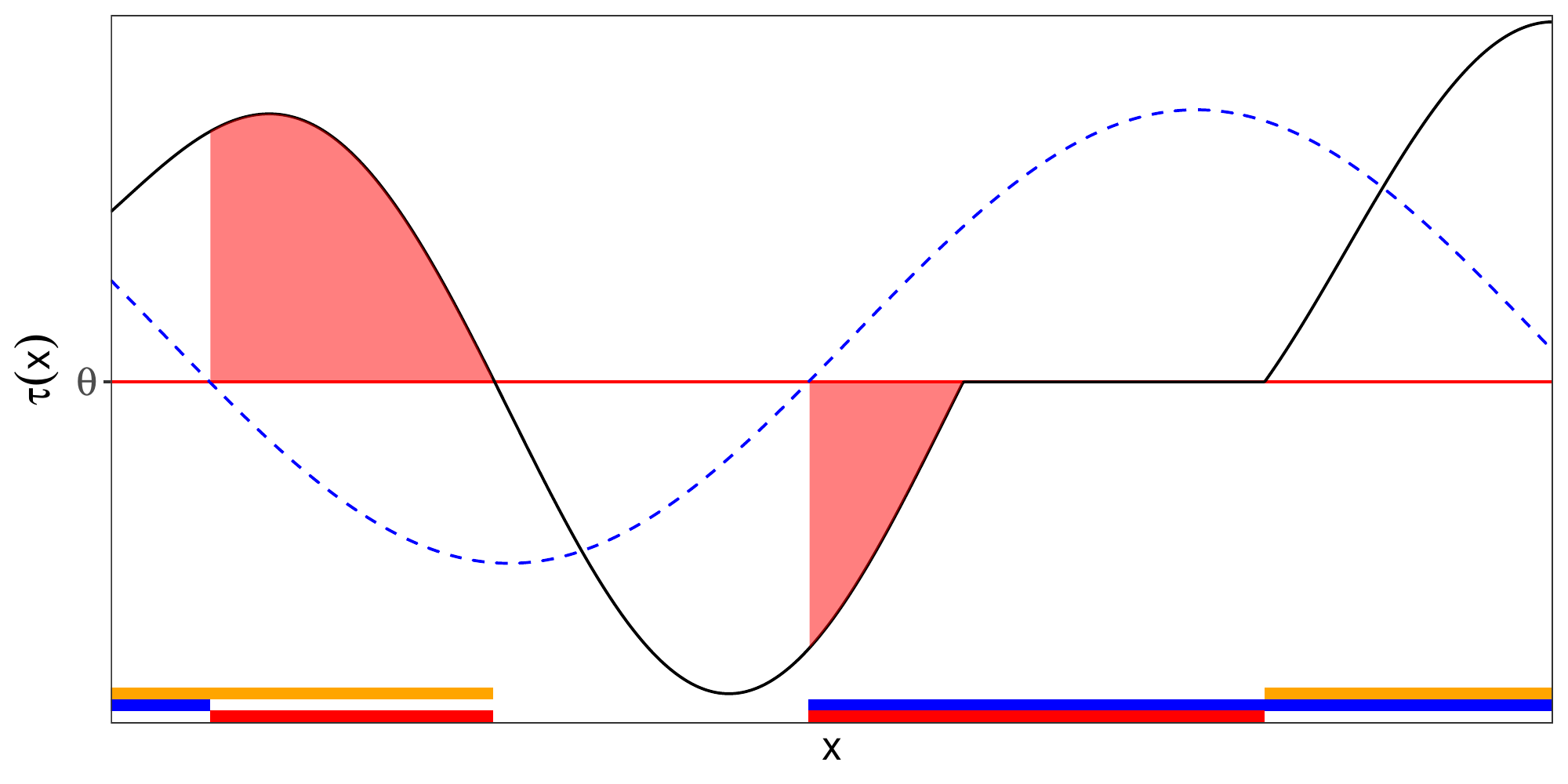}
	\caption{\label{fig:loss} Representation of the loss in eq. \eqref{eq:risk} for $d = 1$ and $X$ uniformly distributed. The solid line is $\tau(x)$, the dotted line is $\widehat\tau(x)$, and the shaded area equals $d_H(\widehat\Gamma, \Gamma)$. The orange portion of the $x$-axis represents $\Gamma(\theta)$, the blue one $\widehat\Gamma(\theta)$ and the red one $\widehat\Gamma \Delta \Gamma$.}
\end{figure}
\begin{remark}\label{remark:willet}
	\cite{willett2007minimax} study estimation of the level sets of a function using dyadic trees. Their approach, adjusted to our setting, would find $\widehat\Gamma(\theta)$ by minimizing an estimate of
	\begin{align*}
		R\{\overline\Gamma(\theta)\} \propto \int \{\theta - \tau(x)\}[\one\{x \in \overline\Gamma(\theta)\} - \one\{x \in \overline\Gamma^c(\theta)\}] f(x) dx
	\end{align*}
	as the risk function. They show that minimizing $R\{\overline\Gamma(\theta)\}$ is equivalent to minimizing the excess risk, i.e.
	\begin{align*}
		R\{\widehat\Gamma(\theta)\} - R\{\Gamma(\theta)\}= \int_{\Gamma \Delta \widehat\Gamma} |\tau(x) - \theta| f(x) dx
	\end{align*}
	which is equivalent to the loss $d_H(\widehat\Gamma, \Gamma)$ that we use in this paper. We leave the study of empirical risk minimizers for estimating CATE level sets for future work. 
\end{remark}
As described below, the performance of our estimators will depend crucially on the difficulty in estimating the CATE around the level $\theta$. The intuition is that, to estimate $\Gamma(\theta)$, one needs to estimate the sign of $\tau(x) - \theta$ well and, in regions of the covariates' space where $\tau(x)$ is far from $\theta$, estimating this sign well does not require estimating $\tau(x)$ precisely. On the contrary, for values of $x$ such that $\tau(x)$ is close to $\theta$, estimating $\tau(x)$ precisely plays an important role in determining the sign of $\tau(x) - \theta$. For example, $\tau(x)$ may be a very complex function far away from $\theta$ but, as long as it is well-behaved and easy to estimate close to $\theta$, one may hope to still be able to estimate $\Gamma(\theta)$ well. In this respect, a typical example that we consider is when $\tau(x)$ is $\gamma$-smooth in a neightborhood around $\theta$ and $\gamma{'}$-smooth everywhere else. 
\section{Estimation}
\subsection{Estimand \& setup}
The goal of this section is to provide an upper bound on the risk \eqref{eq:risk} for generic CATE estimators. Following \cite{rigollet2009optimal}, we introduce a margin assumption governing the mass concentrated around the level set $\chi = \{ x \in \mathcal{X}: \tau(x) = \theta\}$ encoded below. 
\begin{assumption}\label{assumption_margin}
	There exist positive constants $t_0$ and $c_0$ such that such that, for all $t \in (0, t_0]$, it holds that $\Pb_X(0 < |\tau(X) - \theta| < t) \leq c_0t^\xi$ . 
\end{assumption}
The margin condition (Assumption \ref{assumption_margin}) can yield fast convergence rates when the performance is measured by the risk in eq. \eqref{eq:risk}. Crucially, it can be shown to hold with exponent $\xi = 1$ as long as the density of $\tau(X)$ is bounded, which can be satisfied in many applications. The following two propositions are restatements of Lemmas 5.1 and 5.2 in \cite{audibert2007fast} written for the problem considered here; we provide their proofs for completeness. 
\begin{proposition}\label{prop:general}
	Under Assumption \ref{assumption_margin}, it holds that
	\begin{align*}
		\E\left\{d_H(\widehat\Gamma, \Gamma)\right\} \leq \E\left[\int_{\mathcal{X}} \one\left\{|\tau(x) - \theta| \leq \| \widehat\tau - \tau\|_\infty\right\}|\tau(x) - \theta| f(x) dx\right] \lesssim \E\left(\|\widehat\tau - \tau\|_\infty^{1 + \xi}\right)
	\end{align*}
\end{proposition}
\begin{proof}
	The proposition simply follows from the observation that
	\begin{align*}
		\one\left\{x \in \widehat\Gamma(\theta) \Delta \Gamma(\theta) \right\} & = \left| \one\left\{ \widehat\tau(x) - \theta > 0\right\} - \one\left\{\tau(x) - \theta > 0 \right\} \right| \\
		& \leq \one\left\{|\tau(x) - \theta| \leq | \widehat\tau(x) - \tau(x)|\right\} \\
		& \leq \one\left\{|\tau(x) - \theta| \leq \| \widehat\tau - \tau\|_\infty\right\}
	\end{align*}
	The second inequality follows by Lemma 1 in \cite{kennedy2020sharp}. 
\end{proof}
Proposition \ref{prop:general} applies to any estimator $\widehat\tau$ of $\tau$ and links the error in estimating the upper level sets to the error in estimating $\tau$. In particular, it is often the case that $\|\widehat\tau - \tau\|_\infty$ is of the same order of the pointwise error $|\widehat\tau(x) - \tau(x)|$ up to a log factor. In this sense, Proposition \ref{prop:general} would typically match the sharper result described in Lemma \ref{lemma:lemma3.1} up to a log factor provided that estimating $\tau(x)$ near the level $\theta$ is at least as difficult as estimating it on the entire domain. The next proposition links the level set estimator error to the $L_p$ norm of the error in estimating $\tau$. This proposition, however, appears to give results that match those in Lemma \ref{lemma:lemma3.1} only if the margin condition does not hold, i.e. $\xi = 0$.
\begin{proposition}\label{prop:general_pnorm}
	Under Assumption \ref{assumption_margin}, it holds, for any $1 \leq p < \infty$:
	\begin{align*}
		\E\left\{d_H(\widehat\Gamma, \Gamma)\right\} \leq C_{\xi, p} \E\left\{\|\widehat\tau - \tau\|_{p}^{\frac{p(1 + \xi)}{p + \xi}}\right\}
	\end{align*} 
	for some constant $C_{\xi, p}$ depending on $p$ and $\xi$. 
\end{proposition}
\begin{proof}
	It holds that
	\begin{align*}
		d_H(\widehat\Gamma, \Gamma) & \leq \int \one\left\{|\tau(x) - \theta| \leq | \widehat\tau(x) - \tau(x)|\right\} \one\left\{ 0 < |\tau(x) - \theta|\leq t\right\} |\tau(x) - \theta| f(x) dx \\
		& \hphantom{=} + \int \one\left\{|\tau(x) - \theta| \leq | \widehat\tau(x) - \tau(x)| \right\} \one\left\{ |\tau(x) - \theta| > t\right\} |\tau(x) - \theta| f(x) dx \\
		& \leq \int \one\left\{|\tau(x) - \theta| \leq | \widehat\tau(x) - \tau(x)|\right\} \one\left\{ 0 < |\tau(x) - \theta|\leq t\right\} |\tau(x) - \widehat\tau(x)| f(x) dx \\
		& \hphantom{=} + \int \one\left\{|\tau(x) - \theta| \leq | \widehat\tau(x) - \tau(x)| \right\} \one\left\{ |\tau(x) - \theta| > t\right\} |\tau(x) - \widehat\tau(x)| f(x) dx \\
		& \lesssim \|\widehat\tau - \tau\|_p t^{\frac{\xi}{p}(p-1)} + \frac{\|\widehat\tau - \tau\|^p_p}{t^{p-1}}
	\end{align*}
	by H\"{o}lder's inequality. Minimizing the RHS over $t$ yields the desired bound.
\end{proof}
Proposition \ref{prop:general} and \ref{prop:general_pnorm} show that larger values of $\xi$ make estimation of the upper level sets easier. However, as noted in \cite{audibert2007fast}, $\xi$ cannot be too large or else the class of distributions satisfying the margin condition becomes small. This is true, for example, in smoothness models where $\tau(x)$ is $\gamma$-smooth around the cutoff in the sense of Definition \ref{def:holder}. If $\tau(x)$ is smooth enough around the cutoff, it cannot jump away from the level $\theta$ too quickly. This means that the measure of the set where it stays close to the cutoff cannot be too small and thus $\xi$ cannot be too large. In particular, following the proof of Proposition 3.4 in \cite{audibert2007fast}, $\xi\min(1, \gamma) \leq 1$ is necessary for $\tau(x)$ to cross $\theta$ in the interior of the support of the distribution of $X$, when this has a density bounded above and below away from zero. 

The lemma below, which is essentially Lemma 3.1 in \cite{rigollet2009optimal} and Theorem 3.1 in \cite{audibert2007fast} adjusted for our setting, shows that if $\widehat\tau(x) - \tau(x)$ satisfies an exponential inequality, then the bound on the risk $\E\{d_H(\widehat\Gamma, \Gamma)\}$ can be sharpened relative to the results presented in Propositions \ref{prop:general} and \ref{prop:general_pnorm} above. Furthermore, the bound on the risk depends on how fast $\widehat\tau(x)$ converges to $\tau(x)$ for values of $x$ near the cutoff $\tau(x) = \theta$.
\begin{lemma} \label{lemma:lemma3.1}
	Fix $\eta > 0, \Delta > 0$ and let $D(\eta) = \left\{x \in \mathcal{X}: |\tau(x) - \theta| \ \leq \eta \right\}$. Let $a_n$, $b_n$ and $\delta_n$ be monotonically decreasing sequences. Suppose that
	\begin{enumerate}
		\item $b_n \leq c_1 (\log n)^{-1/\kappa_2 - \epsilon}$, with $\epsilon > 0$;
		\item $a_n \geq c_2 n^{-\mu}$ for some positive constants $c_2$ and $\mu$, and $a_n \leq b_n$; 
	\end{enumerate}  
	and that the following inequalities hold
	\begin{align}\label{eq:expo_ineq}
		& \Pb\left(|\widehat\tau(x) - \tau(x)| > t \right) \leq c_3 e^{-c_4 (t / a_n)^{\kappa_1}} + c_5\frac{\delta^{1 + \xi}_n}{t^{1 + \xi}} \\
		&  \text{ for all } x \in D(\eta) \text{ and } c_a a_n < t < \Delta, \text{ and} \nonumber \\
		& \Pb\left(|\widehat\tau(x) - \tau(x)| > t \right) \leq c_6 e^{-c_7(t / b_n)^{\kappa_2}} + c_8 \frac{\delta^{1 + \xi}_n}{t^{1 + \xi}}  \\
		& \text{ for all } x \not\in D(\eta) \text{ and } c_b b_n < t < \Delta \nonumber.
	\end{align}
	for some constants $c_a, c_b, c_1, \ldots, c_8, \kappa_1$, and $\kappa_2$. Then, $\E\{d_H(\widehat\Gamma, \Gamma)\} \lesssim a_n^{1 + \xi} + (c_5 \lor c_8) \delta_n^{1 + \xi} \log n$. In particular, if $c_5 = c_8 = 0$, then $\E\{d_H(\widehat\Gamma, \Gamma)\} \lesssim a_n^{1 + \xi}$. 
\end{lemma}
The central requirement to apply Lemma \ref{lemma:lemma3.1} is that the estimator $\widehat\tau(x)$ must satisfy an exponential inequality. If this is the case, this lemma shows that, provided that $\widehat\tau(x)$ converges to $\tau(x)$ at a rate $(\log n)^{-\kappa}$ for some $\kappa$ on the entire domain, the accuracy for estimating the CATE level set $\Gamma(\theta)$ is entirely determined by the rate for estimating $\tau(x)$ near the level $\theta$. If it is hard to show that the estimator satisfies an exponential inequality, then one can resort to applying Propositions \ref{prop:general} and \ref{prop:general_pnorm}. Because of Lemma \ref{lemma:lemma3.1}, we are left with the task of deriving concentration inequalities for $|\widehat\tau(x) - \tau(x)|$. We will do that for the case when $\widehat\tau(x)$ is either a DR-Learner or an Lp-R-Learner, which may be of independent interest. 
\subsection{Bound on estimation error using a DR-Learner}
To start, we consider the DR-Learner proposed and analyzed by \cite{kennedy2020optimal}.
\begin{definition}[DR-Learner algorithm based on linear smoothing]\label{def:dr_learner}
	Let $D^n$ and $Z^n$ be two independent samples of observations.
	\begin{enumerate}
		\item Using only observations in $D^n$, construct estimators $\widehat\pi(x) = \widehat\Pb(A = 1 \mid X - x)$ and $\widehat\mu_a(x) = \widehat\E(Y \mid A = a, X = x)$.
		\item Using only observations in $Z^n$, construct
		\begin{align*}
			& \widehat\tau(x) = \sum_{i = 1}^n W_i(x; X^n) \widehat\varphi(Z_i), \quad \text{for} \\
			& \widehat\varphi(Z_i) = \frac{\{A- \widehat\pi(X_i)\}\{Y_i - \widehat\mu_A(X_i)\}}{\widehat\pi(X_i)\{1 - \widehat\pi(X_i)\}}+ \widehat\mu_1(X_i) - \widehat\mu_0(X_i),
		\end{align*}
		some weights $W_i(x; X^n)$ and $X^n \subset Z^n$.
	\end{enumerate}
\end{definition}
Let us define:
\begin{align*}
	& S(x; X^n) = \left\{\sum_{i = 1}^n W^2_i(x; X^n)\right\}^{1/2}, \quad \widehat{b}(X_i) = \E\{\widehat\varphi(Z_i) - \varphi(Z_i) \mid X_i\} \quad \text{and}\\
	& \Delta(x; X^n) = \sum_{i = 1}^nW_i(x; X^n)\tau(X_i) - \tau(x)
\end{align*}
The quantity $\Delta(x; X^n)$ is the smoothing bias (conditional on $X_1, \ldots, X_n$) of the oracle estimator that has access to the true function $\varphi(Z_i)$. The quantity $\widehat{b}(x)$ expresses the bias resulting from having to estimate the nuisance functions. If $\pi(x)$ and $\widehat\pi(x)$ are bounded away from zero and one (positivity), it can be shown that
\begin{align*}
	|\widehat{b}(x)| \lesssim \left|\{\pi(x) - \widehat\pi(x)\}[\{\mu_1(x) - \widehat\mu_1(x)\} + \{\mu_0(x) - \widehat\mu_0(x)\}] \right|
\end{align*}
\begin{remark}
	A major advantage of the DR-Learner framework, not necessarily based on linear smoothing, is that regressing an estimate of the pseudo-outcome $\widehat\varphi(Z)$ on $V = v$ yields an estimate of $\tau(v) = \E\{\mu_1(X) - \mu_0(X) \mid V = v\}$, i.e., the CATE function evaluated at effect modifiers $V$, which may differ from the covariates $X$ needed to deconfound the treatment-outcome association. This is particularly useful when the dimension of $X$ is much greater than that of $V$. Thus, if the goal is to compute the upper level sets of $\tau(v)$, with $v \neq x$, thresholding a DR-Learner is an attractive option. 
\end{remark}
We have the following exponential inequality. 
\begin{lemma}\label{lemma:moment_bound_dr}
	Suppose $\widehat\tau(x)$ is a DR-Learner defined in \eqref{def:dr_learner}. Further suppose that
	\begin{enumerate}
		\item $|\Delta(x; X^n)| \leq c_1a_n$ almost surely for a monotonically decreasing sequence $a_n$ and constant $c_1$;
		\item $\E\{S^p(x; X^n)\} \leq s_n^p$ for any $p > 0$, a monotonically decreasing sequence $s_n$;
		\item $\E\left| \sum_{i = 1}^n W_i(x; X^n) \widehat{b}(X_i)\right|^{1 + \xi} \leq \delta_n^{1 + \xi}$ for a monotonically decreasing sequence $\delta_n$;
		\item $\|\widehat\varphi - \varphi - \widehat{b}\|_\infty \leq c_2 \|\varphi\|_\infty$ for a constant $c_2$.
	\end{enumerate}
	Then, for any $t \geq 3c_1a_n$, it holds that
	\begin{align*}
		\Pb\left(| \widehat\tau(x) - \tau(x)| > t \right) \leq 2e^2 \exp\left\{-\left(\frac{t}{12(c_2 \lor 2) e \|\varphi\|_\infty s_n} \right)^2\right\} + 3^{1 + \xi}\left(\frac{\delta_n}{t}\right)^{1 + \xi}
	\end{align*}
	
\end{lemma}
Lemma \ref{lemma:moment_bound_dr} provides an exponential inequality for the DR-Learner based on linear smoothing, which might be of independent interest. Conditions 1-3 are not really assumptions in the sense that they are simply used to state the inequality in a succint form. Depending on the weights of the linear smoother and the accuracy in estimating the nuisance functions, conditions 1-3 would be satisfied by different sequences $a_n$, $s_n$ and $\delta_n$. Condition 4 is a mild boundedness assumption.
In the following example, we show how Lemmas \ref{lemma:lemma3.1} and \ref{lemma:moment_bound_dr} can be used to derive an upper bound on $\E\{d_H(\widehat\Gamma, \Gamma)\}$, where $\widehat\Gamma = \{x \in \mathcal{X}: \widehat\tau(x) > \theta\}$ for $\widehat\tau(x)$ the DR-Learner based on local polynomial second stage regression. 
\begin{example}[DR-Learner with local polynomials]\label{ex:dr_locpoly}
	Suppose that $\tau(x)$ is $\gamma$-smooth locally around any $x \in D(\eta)$ in the sense of Definition \ref{def:holder} and it is $\gamma{'}$-smooth for any $x \not \in D(\eta)$. Further suppose that $\widehat\tau(x)$ is based on local polynomial  second stage regression and that all observations are bounded. That is, $W_i(x; X^n)$ are the weights of a local polynomial of degree $p = \lfloor \gamma \rfloor$. The calculations in \cite{tsybakov2004introduction} (Section 1.6) show that, under mild regularity conditions:
	\begin{align*}
		S(x; X^n) \lesssim \frac{1}{\sqrt{nh^d}}, \quad T(x; X^n) = \sum_{i = 1}^n |W_i(x; X^n)| \lesssim 1, \quad \text{and} \quad |\Delta(x; X^n)| \lesssim h^\gamma.
	\end{align*}
	for $x \in D(\eta)$. Choosing $h$ of order $n^{-1/(2\gamma + d)}$ yields that there exist constants $c_1$ and $c_2$ such that
	\begin{align*}
		S(x; X^n) \leq c_1n^{-\gamma/(2\gamma + d)} \quad \text{and} \quad |\Delta(x; X^n)| \leq c_2n^{-\gamma/(2\gamma + d)}
	\end{align*}
	Typically, it will be the case that $W_i(x; X^n) = 0$ if $\|X_i - x\| > h$ so that by Jensen's inequality (since $u \mapsto |u|^{1 + \xi}$ is convex): 
	\begin{align*}
		\E\left| \sum_{i = 1}^n W_i(x; X^n) \widehat{b}(X_i) \right|^{1 + \xi} & \leq \E\left[T^{1 + \xi}(x; X^n) \cdot \left\{ \frac{\sum_{i = 1}^n |W_i(x; X^n)| \left| \widehat{b}(X_i) \right|}{T(x; X^n)} \right\}^{1 + \xi} \right] \\
		&  \leq  \E\left[T^{1 + \xi}(x; X^n) \cdot \left\{ \frac{\sum_{i = 1}^n |W_i(x; X^n)| \left| \widehat{b}(X_i) \right|^{1 + \xi}}{T(x; X^n)} \right\} \right] \\
		& \leq \E\left\{T^{1 + \xi}(x; X^n)\sup_{u: \|u - x\|\leq h} \left| \widehat{b}(u) \right|^{1 + \xi}\right\} \\
		& = \E \left\{T^{1 + \xi}(x; X^n) \right\} \E\left\{\sup_{u: \|u - x\|\leq h} \left| \widehat{b}(u) \right|^{1 + \xi}\right\} \\
		& \lesssim \E\left\{\sup_{u: \|u - x\|\leq h} \left| \widehat{b}(u) \right|^{1 + \xi}\right\}
	\end{align*}
	where the last equality follows because $\widehat{b}(u)$ depends only on the observations in the training sample, which is independent of $X^n$. By Lemma \ref{lemma:moment_bound_dr} and all $t \geq 3 c_2 n^{-\gamma/(2\gamma + d)}$ and $x \in D(\eta)$:
	\begin{align*}
		& \Pb\left(| \widehat\tau(x) - \tau(x)| > t \right) \lesssim \exp\left(-Ct^2n^{-2\gamma/(2\gamma + d)}\right) \\
		& \quad + t^{-1-\xi} \E\left(\sup_{u: \|u - x\|\leq h} \left|\{\pi(u) - \widehat\pi(u)\}[\{\mu_1(u) - \widehat\mu_1(u)\} + \{\mu_0(u) - \widehat\mu_0(u)\}] \right|^{ 1 + \xi} \right)
	\end{align*}
	For $x \not \in D(\eta)$, we have the same inequality with $\gamma$ replaced by $\gamma'$. Thus, by Lemma \ref{lemma:lemma3.1}, we have
	\begin{align*}
		\E\{d_H(\widehat\Gamma, \Gamma)\} \lesssim n^{-(1 + \xi)\gamma/(2\gamma + d)}+ \delta_n^{1 + \xi} \log n
	\end{align*}
	where $\delta_n$ satisfies
	\begin{align*}
		\E\left(\sup_{u: \|u - x\|\leq h} \left|\{\pi(u) - \widehat\pi(u)\}[\{\mu_1(u) - \widehat\mu_1(u)\} + \{\mu_0(u) - \widehat\mu_0(u)\}] \right|^{ 1 + \xi} \right) \lesssim \delta_n^{1 + \xi}.
	\end{align*}
\end{example}

\subsection{Bound on estimation error using Lp-R-Learners}
In this section, we derive an exponential inequality when $\widehat\tau(x)$ is the Lp-R-Learner. To describe the Lp-R-Learner estimator, we need to introduce some additional notation. We refer to the original paper \cite{kennedy2022minimax} for more details. In particular, the authors consider two different parametrizations of the data generating process: one based on $(f, \pi, \mu_0, \tau)$, which we consider in our work, and one based on $(f, \pi, \mu, \tau)$, where $\mu(x) = \E(Y \mid X = x)$ (see their Section 6). We expect that extending our analysis to cover the latter parametrization is straightforward.
\begin{definition}[Lp-R-Learner]\label{def:lp-r-learner}
	Let $F$ denote the CDF of $X$. For each covariate $x_j$, let $\rho(x_j) = [\rho_0(x_j), \ldots, \rho_{\lfloor \gamma \rfloor}(x_j)]$ be the first $(\lfloor \gamma \rfloor + 1)$ Legendre polynomials shifted to be orthonormal in $[0, 1]$. That is, 
	\begin{align*}
		\rho_m(x_j) = \sum_{l =1}^m \theta_{lm} x_j^l, \text{ for } \theta_{lm} = (-1)^{l + m} \sqrt{2m + 1} {m \choose l} {{m + l} \choose l}
	\end{align*}
	Define $\rho(x)$ to be the tensor product containing all interactions of $\rho(x_1), \ldots, \rho(x_d)$ up to order $\lfloor \gamma \rfloor$. Thus, $\rho(x)$ has length $J = {{d + \lfloor \gamma \rfloor} \choose \lfloor \gamma \rfloor} $ and is orthonormal in $[0, 1]^d$. Finally, define $\rho_h(x) = \rho(0.5 + (x - x_0) / h)$. The Lp-R-Learner $\widehat\tau(x_0)$ is defined as
	\begin{align*}
		& \widehat\tau(x_0) = \rho^T_h(x_0) \widehat{Q}^{-1} \widehat{R}, \text{ where } K_h(x) = \one(2\|x - x_0\|\leq h) \\
		& \widehat{Q} = \Pn \left\{ \rho_h(X) K_h(X) \widehat\varphi_{a1}(Z)\rho^T_h(X)\right\} + \mathbb{U}_n\left\{\rho_h(X_1)K_h(X_1)\widehat\varphi_{a2}(Z_1, Z_2) K_h(X_2)\rho^T_h(X_1) \right\} \\
		& \widehat{R} =\Pn\left\{ \rho_h(X_1) K_h(X)\widehat\varphi_{y1}(Z)\right\} + \mathbb{U}_n\left\{\rho_h(X_1)K_h(X_1)\widehat\varphi_{y2}(Z_1, Z_2)K_h(X_2) \right\} \\
		& \widehat\varphi_{a1}(Z) = A\{A -\widehat\pi(X) \} \\
		& \widehat\varphi_{a2}(Z_1, Z_2) = -\{A_1 - \widehat\pi(X_1)\} K_h(X_1) b^T_h(X_1)\widehat\Omega^{-1} b_h^T(X_2) A_2 \\
		& \widehat\varphi_{y1} (Z)= \{Y - \widehat\mu_0(X)\}\{A - \widehat\pi(X)\} \\
		& \widehat\varphi_{y2}(Z_1, Z_2) = -\{A_1 - \widehat\pi(X_1)\} b^T_h(X_1)\widehat\Omega^{-1} b_h^T(X_2) \{Y_2 - \widehat\mu_0(X_2)\} \\
		& b_h(X) = b(0.5 + (x - x_0) / h)\one(2\|x - x_0\| \leq h) \\
		& \widehat\Omega  = \int_{v \in [0, 1]^d}b(v)b^T(v) d\widehat{F}(x_0 + h(v - 0.5))
	\end{align*}
	for $b:\R^d \mapsto \R^k$ a basis vector of dimension $k$ that should have good approximating properties for the nuisance function class. The nuisance functions $(\widehat{F}, \widehat\pi, \widehat\mu_0)$ are computed from a training sample $D^n$, independent of that used to calculate the empirical and $U$-statistic measures.
\end{definition}
The Lp-R-Learner estimator is tailored to a particular smoothness model, which we describe next and adopt in this section and when discussing minimax optimality.
\begin{definition}[Lp-R-Learner smoothness model]\label{def:smoothness}
	Fix $\gamma$, $\gamma{'}$, $\alpha$ and $\beta$. Recall that $D(\eta) = \{x \in \mathcal{X}: |\tau(x) - \theta| \leq \eta\}$, $\eta > 0$. We define $\mathcal{P}$ to be the collection of all distributions satisfying the following conditions:
	\begin{enumerate}
		\item $\tau(x)$ is $\gamma$-smooth locally around any $x \in D(\eta)$ in the sense of Definition \ref{def:holder};
		\item $\tau(x)$ is $\gamma{'}$-smooth for any $x \not \in D(\eta)$;
		\item $\mu_0(x) - \widehat\mu_0(x)$ is $\beta$-smooth and $\pi(x) - \widehat\pi(x)$ is $\alpha$-smooth for any $x \in \mathcal{X}$ \footnote{In principle, $\mu_0(x)$ and $\pi(x)$ could have different smoothness levels depending on whether $x \in D(\eta)$ or not. This would not complicate the analysis conceptually but it would make the notation more involved. For simplicity, we treat the nuisance functions as having a smoothness level that does not vary across the covariates' space.};
		\item $\epsilon \leq \pi(x) \leq 1 - \epsilon$ almost-surely, for some $\epsilon > 0$;
		\item The eigenvalues of $Q$ and $\Omega$ are bounded above and below away from zero.
	\end{enumerate}
\end{definition}
Let $s = (\alpha + \beta)/2$ denote the average smoothness of the nuisance functions. Let $T = 1 + d/ (4s) + d/(2\gamma)$ and $T' = 1 + d/ (4s) + d/(2\gamma')$. For the model described in Definition \ref{def:smoothness},  \cite{kennedy2022minimax} proved that, under certain regularity conditions, the pointwise risk satisfies
\begin{align*}
	\E|\widehat\tau(x) - \tau(x)|^2 \lesssim \begin{cases}
		n^{-2\gamma/(2\gamma + d)} & \text{if } x \in D(\eta) \text{ and } s \geq \frac{d/4}{1 + d/(2\gamma)} \\
		n^{-2/T} & \text{if } x \in D(\eta) \text{ and } s < \frac{d/4}{1 + d/(2\gamma)} \\
		n^{-2\gamma'/(2\gamma' + d)} & \text{if } x \not\in D(\eta) \text{ and } s \geq \frac{d/4}{1 + d/(2\gamma')} \\
		n^{-2/T'} & \text{if } x \not\in D(\eta) \text{ and } s < \frac{d/4}{1 + d/(2\gamma')} 
	\end{cases}
\end{align*}
Crucially, \cite{kennedy2022minimax} shows that these rates are the minimax optimal rates for estimating $\tau(x)$ in this model. Notice that the rate $n^{-2\gamma/(2\gamma + d)}$ is the optimal rate for estimating a $d$-dimensional, $\gamma$-smooth regression function. It is referred to as the \textit{oracle rate} because it is the fastest rate achievable by an infeasible estimator that has access to the true pseudo-outcomes $\varphi(Z_i)$ (see definition \ref{def:dr_learner}). In the next section, we show that $\widehat\Gamma(\theta)$ based on thresholding the Lp-R-Learner estimator of the CATE is minimax optimal for $\Gamma(\theta)$ in this model as well. We derive the following exponential inequality, which may be of independent interest. 
\begin{lemma}\label{lemma:moment_lp-r-learner}
	Suppose the data generating mechanism satisfies the model described in Definition $\ref{def:smoothness}$. Let $dF^*(v) = dF(x_0 + h(v - 0.5))$ and $\|g\|^2_{F^*} = \int g^2(v) dF^*(v)$. Further suppose that:
	\begin{enumerate}
		\item The quantities $y^2$, $\widehat\pi^2$, $\widehat\mu_0^2$,  $\|\mu_0 - \widehat\mu_0\|_{F^*}$, $\|\widehat{Q}^{-1} - Q^{-1}\|$ are all bounded above and $\| dF / d\widehat{F} \|_\infty$, $\|\widehat{Q}\|$ and $\|\widehat\Omega\|$ are bounded above and below away from zero;
		\item $\| dF / d \widehat{F} - 1\|_\infty \left\{\|\widehat\pi - \pi\|_{F^*}(h^\gamma + \|\widehat\mu_0 - \mu_0\|_{F^*})\right\}\lesssim n^{-1/T} \lor n^{-1/T'} \lor  n^{-1/(1 + d/(2\gamma))}$;
		\item The basis dictionary is suitable for approximating H\"{o}lder functions of order $s$ in the sense that
		\begin{align*}
			\left\| g - b^T\Omega^{-1} \int b(u) g(u) dF^*(u) \right\|_{F^*} \lesssim k^{-s/d}
		\end{align*}
		if $g$ is $s$-smooth. 
	\end{enumerate} 
	Then there exist some constants $C, c, c_r$ and $\Delta$ so that, for all $c_r r_n \leq t \leq \Delta$, it holds that
	\begin{align*}
		\Pb\left(|\widehat\tau(x) - \tau(x)| > t\right) \leq
		C\exp\left[ - c \min \left\{ \left(\frac{t}{r_n} \right)^2, \left(\frac{t}{r_n} \right)^{1/2}\right\} \right],
	\end{align*}
	where
	\begin{align*}
		r_n = \begin{cases}
			n^{-\gamma/(2\gamma + d)} & \text{if } x \in D(\eta) \text{ and } s \geq \frac{d/4}{1 + d/(2\gamma)} \\
			n^{-1/T} & \text{if } x \in D(\eta) \text{ and } s < \frac{d/4}{1 + d/(2\gamma)} \\
			n^{-\gamma'/(2\gamma' + d)} & \text{if } x \not\in D(\eta) \text{ and } s \geq \frac{d/4}{1 + d/(2\gamma')} \\
			n^{-1/T'} & \text{if } x \not\in D(\eta) \text{ and } s < \frac{d/4}{1 + d/(2\gamma')} 
		\end{cases}
	\end{align*}
\end{lemma}
All the conditions listed in the lemma above are needed in the derivation of the convergence rate (in pointwise RMSE) of $\widehat\tau(x_0)$ as proven in \cite{kennedy2022minimax}. We note that condition 2 requires estimating the covariates' density sufficiently well. We refer the reader to the original paper for a detailed discussion of their interpretation. Next, we use Lemmas \ref{lemma:lemma3.1} and \ref{lemma:moment_lp-r-learner} to derive a bound on $\E\{d_H(\widehat\Gamma, \Gamma)\}$ when $\widehat\Gamma$ is estimated using the Lp-R-Learner. 
\begin{corollary}\label{cor:lp-r-learner-bound}
	Under the setup of Lemma \ref{lemma:moment_lp-r-learner}, it holds that $\E\{d_H(\widehat\Gamma, \Gamma)\} \lesssim r_n^{*1+\xi}$, where 
	\begin{align*}
		r^*_n = \begin{cases}
			n^{-\gamma/(2\gamma + d)} & \text{if } s \geq \frac{d/4}{1 + d/(2\gamma)} \\
			n^{-1/T} & s < \frac{d/4}{1 + d/(2\gamma)}
		\end{cases}
	\end{align*}
\end{corollary}
\begin{proof}
	It is sufficient to apply Lemma \ref{lemma:lemma3.1} with $c_5 = c_8 = 0$. 
\end{proof}
In the next section, we show that $r_n^{*1+\xi}$ is also the minimax rate for estimating the level set $\Gamma(\theta)$ in the model described by Definition \ref{def:smoothness} when the risk is $\E\{d_H(\widehat\Gamma, \Gamma)\}$. We will derive the lower bound on the minimax risk in the low-smoothness regime ($s < \frac{d/4}{1 + d/(2\gamma)}$), with the understanding that a similar construction yields the appropriate lower bound in the high-smoothness regime ($s \geq \frac{d/4}{1 + d/(2\gamma)}$).
\section{Minimax lower bound}
Here, the goal is to find a tight lower bound on the minimax risk, defined as:
\begin{align*}
	\inf_{\widehat\Gamma} \sup_{p \in \mathcal{P}} \E_p\{d_H(\widehat\Gamma, \Gamma_p)\} = \inf_{\widehat\Gamma} \sup_{p \in \mathcal{P}} \E_p \left\{\int_{\widehat\Gamma \Delta \Gamma_p}  |\tau_p(x) - \theta| f_p(x) dx \right\}
\end{align*}
where $\mathcal{P}$ is a set of distributions compatible with our assumptions. Calculating the minimax risk for estimating a given parameter is important for at least two reasons. First, it serves as a benchmark for comparing estimators. In particular, if the lower bound on the minimax risk matches the rate of an available estimator, then one can conclude that there is not another estimator that can improve upon the minimax optimal one, at least in terms of a worst-case analysis, without introducing additional assumptions. Conversely, if there are no estimators attaining a rate that matches the minimax lower bound, then one has to either construct a better estimator or tighten the upper or lower bound. In our setting, we show that a valid lower bound matches the upper bound of Corollary \ref{cor:lp-r-learner-bound} up to constants, which therefore establishes the minimax rate for estimating $\Gamma$ under the loss $d_H(\widehat\Gamma, \Gamma)$ in model \ref{def:smoothness}. Furthermore, a tight minimax lower bound is helpful because it precisely characterizes the difficulty in estimating this parameter. 
\begin{theorem}\label{thm:lower_bound}
	Suppose that $\xi \gamma \leq d$. Under assumption \ref{assumption_margin} and the smoothness model defined in \ref{def:smoothness}, then
	\begin{align*}
		\inf_{\widehat\Gamma} \sup_{p \in \mathcal{P}} \E_p\{d_H(\widehat\Gamma, \Gamma_p)\} \gtrsim r_n^{*1+\xi}, \text{ where } r_n^* = n^{-1/T} \text{ and } T = 1 + d/ (4s) + d/(2\gamma).
	\end{align*}
	when $s < \frac{d/4}{1 + d/(2\gamma)}$. 
\end{theorem}
As shown in \cite{kennedy2022minimax}, the rate $r_n^*$ is the minimax rate for estimating $\tau(x)$ (at a point and under the square loss) in the smoothness model encoded in Definition \ref{def:smoothness} in the low smoothness regime. Our result shows that the same estimator can be thresholded to yield an optimal estimator of the CATE level sets. The result in Theorem \ref{thm:lower_bound} aligns with that of \cite{rigollet2009optimal}, where $r_n^*$ is replaced by the optimal minimax rate for estimating a $\gamma$-smooth density on a $d$-dimensional domain, i.e. $n^{-\gamma / (2\gamma + d)}$ (on the root-mean-square error scale).

The proof of Theorem \ref{thm:lower_bound} combines the construction of \cite{rigollet2009optimal}, \cite{kennedy2022minimax} and Assouad’s lemma (specifically, we rely on Theorem 2.12 in \cite{tsybakov2004introduction}). To derive a lower bound on the risk of an estimator, one needs to construct two worst-case distributions $Q_1$ and $Q_2$ such that $Q_1$ and $Q_2$ are similar enough so that one cannot perfectly determine whether a sample is from $Q_1$ or $Q_2$ but, at the same time, the value of the parameter at $Q_1$ is maximally separated from that at $Q_2$. To construct $Q_1$ and $Q_2$ one typically carefully designs fluctuations around the quantities that need to be estimated, in our case $\pi(x)$, $\mu_0(x)$ and $\tau(x)$. As shown in Figure \ref{fig:bumps}, we place bumps on these functions of particular heights depending on the level of smoothness. Our construction extends that of \cite{kennedy2022minimax}, which is localized in a neighborhood around $x = x_0$, to the entire domain of $X$. In particular, it can be used to show that the rate obtained in \cite{kennedy2022minimax} for the pointwise risk is also the minimax rate for the integrated risk $\int \{\widehat\tau(x) - \tau(x)\}^2 dF(x)$, which might be of independent interest. We refer to \cite{kennedy2022minimax} for additional details. Finally, the lower bound from Theorem \ref{thm:lower_bound} applies only to the case $\xi \gamma \leq d$. This condition also appears in the work of \cite{audibert2007fast} and the more stringent condition $\xi \gamma \leq 1$ appears in the lower bound construction of \cite{rigollet2009optimal}. To the best of our knowledge, deriving a tight lower bound without this condition is still an open problem. 
\begin{figure}[!h]
	\centering
	\includegraphics[scale=0.5]{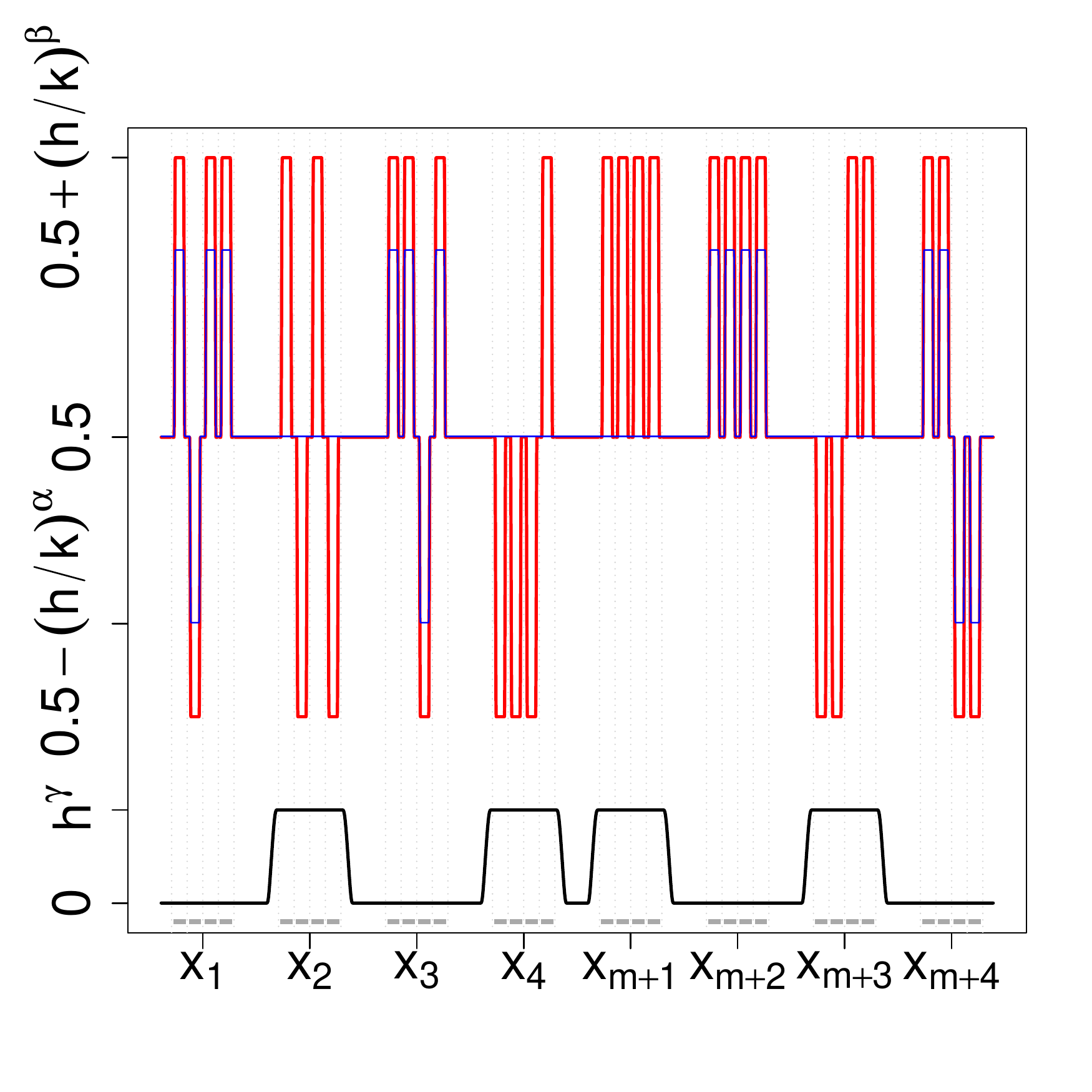}
	\caption{\label{fig:bumps} Lower bound construction for the case $d = 1$, $\theta = 0$ and $\alpha \geq \beta$. The solid black curve represents $\tau(x)$, the red curve represents $\mu_0(x)$ while the blue curve represents $\pi(x)$. Notice that if $\tau(x) > 0$ then $\pi(x) = 1/2$, whereas $\mu_0(x)$ is always fluctuated.}
\end{figure}
\section{Inference} \label{section:inference_lvl_sets}
In this section, we discuss a simple way to carry out inference when a DR-Learner is thresholded to estimate $\Gamma(\theta)$. Inspired by \cite{mammen2013confidence}, we propose constructing two sets $\widehat{C}_l$ and $\widehat{C}_u$ of the form 
\begin{align*}
	& \widehat{C}_l = \left\{x \in \R^d: \widehat\sigma^{-1}(x)\{\widehat\tau(x) - \theta\} > c_n(1-\alpha)\right\} \\
	& \widehat{C}_u = \left\{x \in \R^d: \widehat\sigma^{-1}(x)\{\widehat\tau(x) - \theta\} \geq -c_n(1-\alpha) \right\},
\end{align*}
where $\widehat\sigma(x)$ is an estimate of the standard deviation of $\widehat\tau(x)$ and $c_n(1 - \alpha)$ is some carefully chosen cutoff, depending on the $1-\alpha$ confidence level. The rationale for constructing such sets is outlined in the following lemma, which is written for level sets of some arbitrary function $f(x)$.
\begin{lemma}\label{lemma:inclusions}
	Let $\overline\Lambda(\theta) = \{x \in \mathcal{X}: f(x) \geq \theta\}$ and $\Lambda(\theta) = \{x \in \mathcal{X}: f(x)> \theta\}$. Let $\widehat{f}(x)$ be an estimator of $f(x)$ with some standard deviation $\widehat\sigma(x)$. Define the t-statistic $t_n = \{\widehat\sigma(x)\}^{-1}\{\widehat{f}(x) - f(x)\}$. Finally, define
	\begin{align*}
		& \widehat{C}_l = \{x \in \R^d:\{\widehat\sigma(x)\}^{-1}\{\widehat{f}(x) - \theta\} > t \}  \\
		& \widehat{C}_u = \{x \in \R^d:\{\widehat\sigma(x)\}^{-1}\{\widehat{f}(x) - \theta\} \geq -t \}
	\end{align*}
	Then, it holds that
	\begin{align*}
		\Pb\left( \overline\Lambda(\theta) \subseteq \widehat{C}_u \  \text{ and } \ \widehat{C}_l \subseteq \Lambda(\theta) \right) \geq \Pb\left(\|t_n\|_\infty \leq t  \right). 
	\end{align*}
\end{lemma}
\begin{proof}
	Let $x_0$ be any member of $\overline\Lambda(\theta)$ and notice the following chain of implications
	\begin{align*}
		\|t_n\|_\infty \leq t \implies \widehat\sigma^{-1}(x_0)\{\widehat{f}(x_0) - f(x_0)\}  \geq -t \implies  \widehat\sigma^{-1}(x_0)\{\widehat{f}(x_0) - \theta\}  \geq -t
	\end{align*}
	because $f(x_0) \geq \theta$. This means that $x_0 \in \widehat{C}_u$ so that we conclude that $\Pb\left( \overline\Lambda(\theta) \subseteq \widehat{C}_u \right) \geq \Pb\left(\|t_n\|_\infty \leq t  \right)$.
	
	Similarly, let $x_0$ be any member of $\widehat{C}_l$ and notice that
	\begin{align*}
		& \|t_n\|_\infty \leq t \implies \widehat\sigma^{-1}(x_0)\{\widehat{f}(x_0) - \theta\} +  \widehat\sigma^{-1}(x_0)\{\theta - f(x_0)\}  \leq t \\
		& \implies \widehat\sigma^{-1}(x_0)\{\theta - f(x_0)\} < 0
	\end{align*}
	because $\widehat\sigma^{-1}(x_0)\{\widehat{f}(x_0) - \theta\} > t$. Thus, $x_0 \in \Lambda(\theta)$ so that $$\Pb\left(\widehat{C}_l \subseteq \Lambda(\theta) \right) \geq \Pb\left(\|t_n\|_\infty \leq t  \right).$$
\end{proof} 
In light of Lemma \ref{lemma:inclusions}, $\widehat{C}_l$ and $\widehat{C}_u$ act as $1-\alpha$ lower and upper confidence sets for $\Lambda(\theta)$ as long as $\|t_n\|_\infty \leq t$ with probability at least $1-\alpha$. Thus, constructing $\widehat{C}_l$ and $\widehat{C}_u$ to cover $\Gamma(\theta)$ effectively reduces to the problem of constructing uniform confidence bands around $\widehat\tau(x)$.

Constructing confidence regions for level sets based on the supremum of the function that is being thresholded is an example of confidence sets based on ``vertical variation." An alternative route would be to construct confidence regions based on ``horizontal variation," an example of which would be a confidence region based on approximating the distribution of the Hausdorff distance between the estimated set and the true set. We leave this for future work and refer to \cite{qiao2019nonparametric} and \cite{chen2017density} for more details regarding the differences between these approaches in the context of density estimation. 

\cite{semenova2021debiased} establish uniform confidence bands for a DR-Learner estimator of the CATE such that the second-stage regression is carried out via orthogonal series regression. One can therefore leverage their results (Theorem 3.5) to construct confidence sets for the CATE level sets based on Lemma \ref{lemma:inclusions}. \footnote{Estimating the quantile of $\|t_n\|_\infty$ typically requires that the smoothing bias for estimating the CATE converges to zero faster than the standard error (e.g., see condition (iv) in Theorem 3.5 in \cite{semenova2021debiased}). This condition can be challenging to guarantee in applications, but we note that it is not required if one changes the target of inference to upper level sets of a  ``smoothed version" of the CATE function, i.e. a modified CATE function that can be estimated without smoothing bias. See also \cite{chen2017density}.} Finally, in the context of dose-response estimation, \cite{takatsu2022debiased} construct uniformly valid confidence bands for second-stage local linear smoothers where the outcome is estimated in a first-step. We expect their results to be useful in the setting considered here as well. We plan on including a more precise result on uniform inference for DR-Learners in an updated version of this work. 
\section{Simulation}
The goal of this section is to evaluate the performance of the estimators and investigate the role of various aspects of the data generating processes in finite samples. First, we study the impact of the nuisance functions' estimation step on the coverage of the CATE upper level set. Our estimator of the upper level set will consist of thresholding a DR-Learner estimator of the CATE based on a parametric second-stage linear regression. Based on Lemma \ref{lemma:moment_bound_dr} and Example \ref{ex:dr_locpoly}, we expect the performance of our estimator to deteriorate significantly when the product of the nuisance functions' error is greater than $o_\Pb(n^{-1/2})$. 

Next, we investigate the impact of the parameters governing the margin assumption \ref{assumption_margin} on the error in estimating the upper level set. To simplify the simulation settings, we consider a smooth CATE with bounded density so that $\xi = 1$ in Assumption \ref{assumption_margin} holds, and we increase the constant $c_0$ on the right hand side of the margin condition inequality. We expect that the larger the region of the covariates' space where the CATE is close to the threhsold the harder the estimation problem becomes.

In all simulation scenarios, we define $\theta = 0$, the sample size $n = 1000$, the number of bootstrap replications used in constructing the confidence regions $B = 10^5$ and the number of simulations $I = 500$. We enforce consistency by setting $Y = AY^1 + (1 - A)Y^0$. We approximate the space $[-1, 1]^2$ by a grid of points $x_1, \ldots, x_m$, which are equally spaced points $(x_{1i}, x_{2j})$ for $1 \leq i, j \leq 50$. Uniform coverage is computed relative to this approximation. 

\textbf{Setup 1A: Impact of the nuisance functions' estimation step}. We generate data from the following model:
\begin{align*}
	& X_i \stackrel{iid}{\sim} \text{Unif}(-1, 1), \quad A \mid X_1, X_2 \sim \text{Bin}(\expit(-1 + X_1 + X_2)), \\
	& Y^1 \mid X_1, X_2 \sim N(0.15 - X_1 -0.5 X_1^2 + X_2, 1), \quad Y^0 \mid X_1, X_2 \sim N(0, 1).
\end{align*}
Notice that $\tau(x) = 0.15 - x_1 -0.5 x_1^2 + x_2$, which we assume is correctly specified in the second-stage regression. However, we construct the nuisance functions estimators $\widehat\pi$, $\widehat\mu_a$ by injecting Gaussian noise of order $n^{-1/c}$, for $c = \{0, 2, 3, 3.8, 4, 5\}$ in the true functions. For example, $\widehat\pi(x) = \expit(x^T\widehat\beta)$, where $\widehat\beta = [-1 \ 1 \ 1]^T + \mathcal{N}_3(n^{-1/c}, n^{-1/c}I_3)$. The case $c = 0$ refers to the case where we do not inject any noise. Figure \ref{fig:setup} represents the simulation setup; the black solid line denotes the set of covariates' values where the CATE is zero. As shown in Figure \ref{fig:cvg}, in agreement with our theoretical results, the coverage of the CATE upper level sets starts to deteriorate as soon as the product of the  errors in estimating the nuisance functions equals or exceeds the rate $n^{-1/4}$. 

\textbf{Setup 1B: Impact of the parameters governing the margin assumption \ref{assumption_margin}}.
We generate data as in Setup 1A except that $Y^1 \mid X_1, X_2 \sim N(\kappa(0.15 - X_1 - 0.5X_1^2 + X_2), 1)$, where $\kappa = \{0.1, 0.5, 1, 5, 10\}$. The parameter $\kappa$ is meant to govern the size of the set $\{ x \in \mathcal{X}: |\tau(x)|\  \leq \epsilon\}$ for some fixed $\epsilon$; the smaller $\kappa$ the larger this set is. We thus expect the performance of our estimator to deteriorate as $\kappa$ decreases.  To isolate the impact of varying $\kappa$ on the performance of the estimators, we use the true nuisance functions, instead of the estimated ones, in the construction of the pseudo outcome. In other words, we gauge the impact of $\kappa$ on the oracle estimator. We compute a monte-carlo approximation to $d_H(\widehat\Gamma, \Gamma)$.

As expected, Figure \ref{fig:error} shows simulation evidence that the estimation error as measured by the risk $\E\{d_H(\widehat\Gamma, \Gamma)\}$ decreases, i.e. the estimation problem becomes easier, if the size of covariates' space where the CATE is close to the level decreases. This too is in agreement with the results from the previous sections. 
\begin{figure}
	
	\centering
	\subfloat[True CATE (from Setup A, $\kappa = 1)$ with corresponding level set $\chi(0) = \{x \in {[-1, 1]}^2: \tau(x) = 0\}$ (black line).]{\label{fig:setup}
		\centering
		\includegraphics[width = 0.3\linewidth]{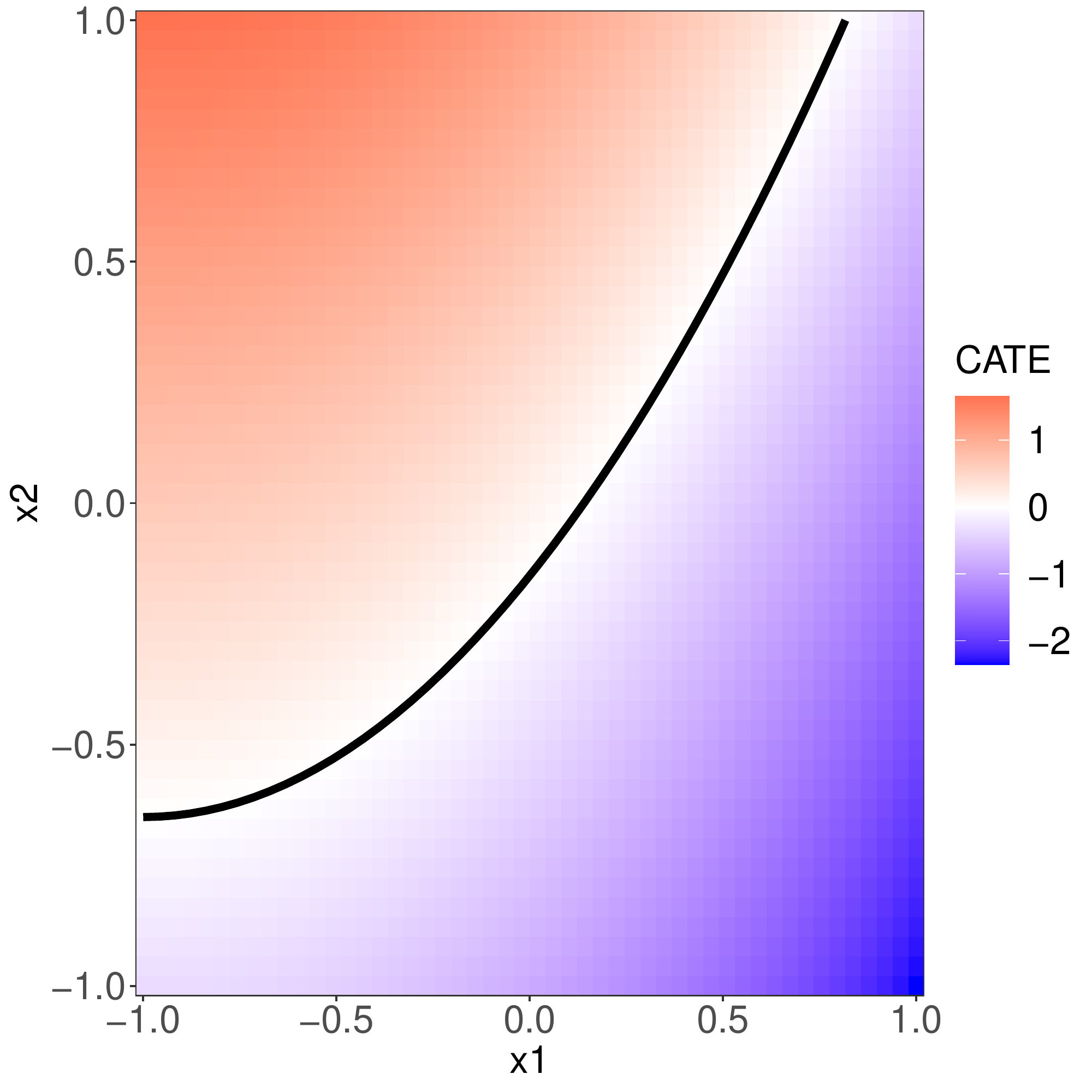}
	}
	\hfill
	\subfloat[Empirical coverage of $\Gamma(0)$ under setup A as a function the nuisance estimators' accuracy ($c$ in $n^{-c}$).]{\label{fig:cvg}
		\centering
		
		\includegraphics[width = 0.3\linewidth]{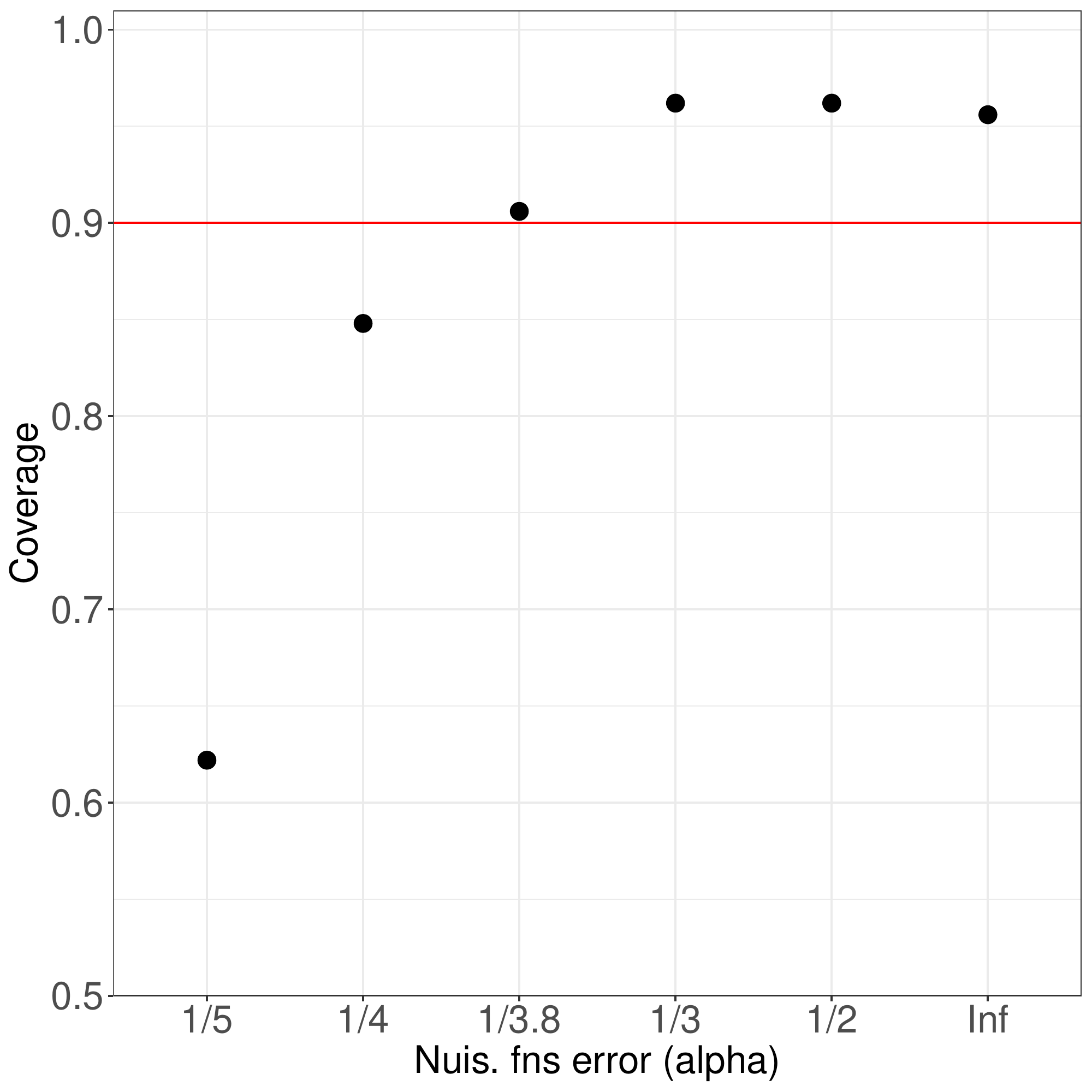}
	}
	\hfill
	\subfloat[Value of the loss $d_H(\widehat\Gamma, \Gamma)$ as a function of $\kappa$ in Setup B.]{\label{fig:error}
		\centering
		
		\includegraphics[width = 0.3\linewidth]{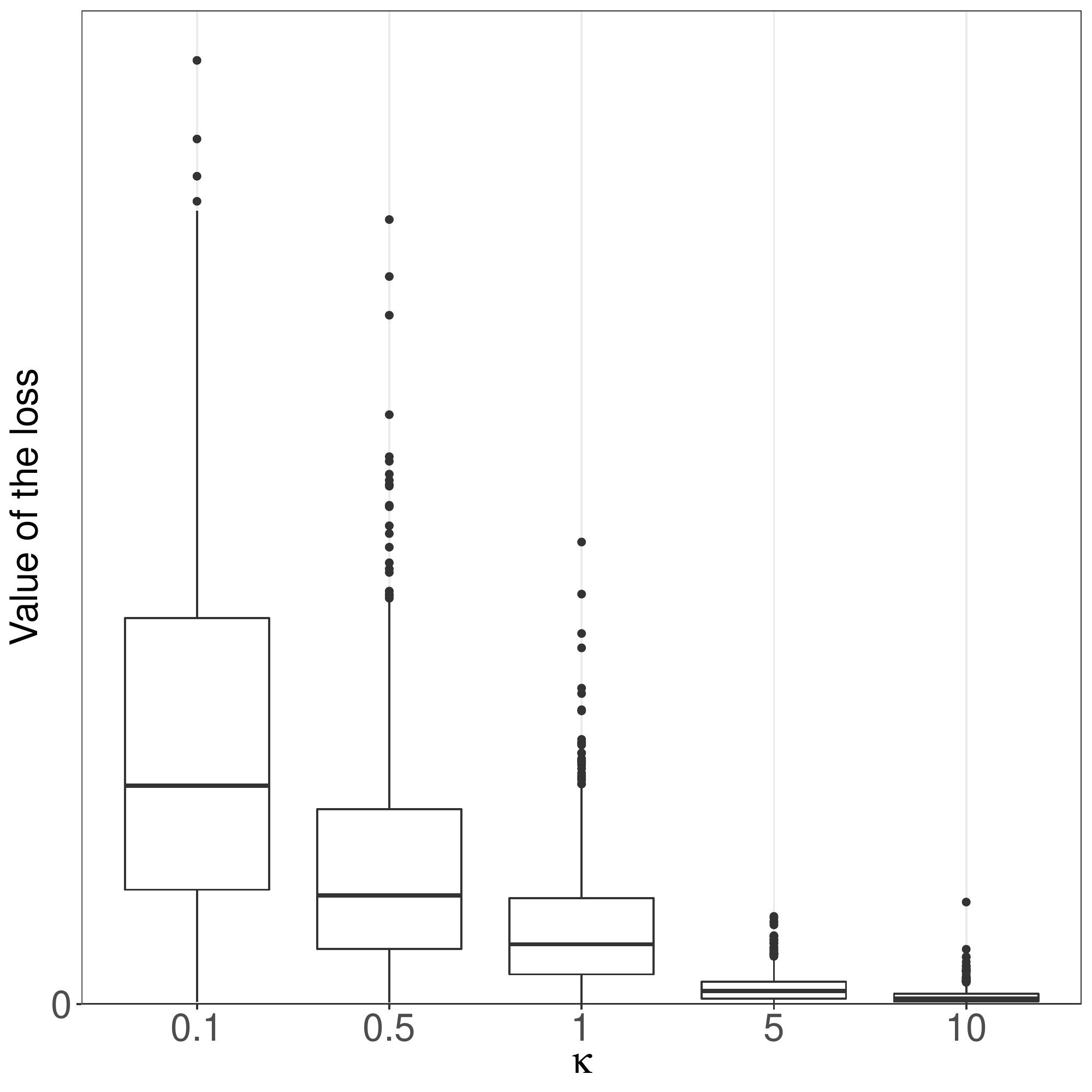}
	}
	
	\caption{Simulation results}
	\label{fig:sim_res}
\end{figure}
\section{Data Analysis}
Partial colon removal, also known as partial colectomy, is a medical procedure where a surgeon removes the diseased portion of the patient's colon and a small portion of surrounding healthy tissue. A partial colectomy serves as a treatment for various conditions including Crohn's disease, ulcerative colitis, and colon cancer. Surgery for appendicitis can be done in two ways. Traditionally, partial colectomy is done via open surgery (OS), which requires a long incision in the abdomen to gain access to the colon. The primary alternative to open surgery is laparoscopic surgery (LS) which is a surgical technique that uses a small incision and small narrow tubes. The surgeon pumps carbon dioxide through the tubes to inflate the organs and create more space for the procedure. Surgical instruments are inserted and used to remove part of the colon. LS is a minimally invasive colectomy and is designed to help patients recover more quickly and experience fewer surgical complications. 

LS for partial colectomy has been widely evaluated in randomized controlled trials, observational studies and meta-analyses \citep{varela2008outcomes,wu2022comparison,wu2010role,kemp2008outcomes,kannan2015laparoscopic}. Across these various types of studies, results indicate that LS leads to better patient outcomes including lower morbidity and lower complications. However, it is also likely that the effect of LS varies from patient to patient. More specifically, there may be some patients for whom LS is particularly beneficial, and there may be other patients for whom it is harmful or ineffective. As an empirical application, we use level sets to characterize optimal treatment for LS for partial colectomy. In our analysis, we use a large observational data set and exploit the large sample size and rich set of covariates to better detect whether the effects of LS vary systematically with key patient characteristics. 

We use a data set that merges the American Medical Association (AMA) Physician Masterfile with all-payer hospital discharge claims from New York, Florida and Pennsylvania in 2012-2013. The data include patient sociodemographic and clinical characteristics including indicators for frailty, severe sepsis or septic shock, and 31 comorbidities based on Elixhauser indices \citep{elixhauser1998comorbidity}. The data also include information on insurance type. Our primary outcomes are indicator variables for mortality and complications. In our data, there are 46,506 patients that underwent a partial colectomy. Among these patients, 20,133 underwent LS and 26,373 underwent OS. 

In Figure \ref{fig:da}, we report the results from a data analysis we have conducted. The figure shows our DR-Learner estimate of the CATE function defined in terms of two effect modifiers, an aggregate measure of comorbidity and age. Age is approximately continuous ranging from 18 to 102 years old, whereas the measure of comorbidity is ordinal taking values in $0, 1, \ldots, 7, 8+$. The outcome is a binary indicator for whether a set of complications occured. To create Figure \ref{fig:da}, we restrict the range of age to be between 30 and 80. In the rest of the age space, we observe too few data points for certain comorbidity levels. We deconfound the treatment / outcome association using all pre-treatment variables available. To estimate the nuisance functions, we use Random Forests implemented in the \texttt{ranger} R package with default parameters. We then estimate the CATE with a DR-Learner (Definition \ref{def:dr_learner}) by regressing the estimated pseudo-outcome on the effect modifiers via a linear model of the form
\begin{align*}
	\widehat\E\left\{\widehat{\varphi}(Z)\mid \text{Age} = v_1, \text{Comorb} = v_2\right\} = \sum_{j = 0}^8 \one(v_2 = j)\left(\widehat\beta_{0j} + \sum_{i = 1}^{k_j} \widehat\beta_{ij} v_1^i\right),
\end{align*}
where $k_j = 1, 2, 3$ is picked by leave-one-out cross validation (loocv). In other words, within each level of comorbidity, we consider a polynomial regression model where the degree, either 1, 2 or 3, is chosen to minimize a leave-one-out estimate of the risk. We construct the confidence regions using the method of approximating the distribution of $\sup_{x}|\widehat\tau(x) - \tau(x)|$ described in \cite{semenova2021debiased}.

Our estimate of the average treatment effect is $-8.07\%$ ($95\%$ CI: $[-8.88\%, \ -7.26\%]$); it is calculated by simply averaging the estimated pseudo-outcomes $\widehat\varphi(Z)$, as those are also the influence-function values for the average treatment effect parameter. Thus, if everyone in the population would undergo laparoscopic surgery instead of the traditional open surgery, we expect a statistically significant reduction on the probability of incurring in complications of roughly 8\%.  This is consistent with the idea that laparoscopic surgery, being minally invasive, reduces the risk to develop complications. 

Importantly, as shown in Figure \ref{fig:cate}, our estimates of the CATE are negative everywhere in the covariates' space considered. From this analysis, and in particular from Figure \ref{fig:cate_inf}, it appears that laparoscopic surgery significantly decreases the chance of complications for units with a low or average number of comorbities across many age groups, and for units roughly older than 44 (52) with a level of comorbidity equal to 6 (7). We do not find any significant decrease in the likelihood of complications for units with an elevated number of comorbidities (8+). We note that the lack of statistical significance for younger units with high levels of comorbidities could be due to the fact that we observe fewer points in these regions. Finally, notice that the blue region in Figure \ref{fig:cate_inf} makes up the complement of $\widehat{C}_u$ (with $\theta = 0$) as defined in Section \ref{section:inference_lvl_sets}. Therefore, their union is a region that, with high probability, is contained in, and thus potentially smaller than, the true region where the CATE is negative.
\begin{figure}
	
	\centering
	\subfloat[Estimates of the CATE values for patients that underwent LS for a partial colectomy, as a function of age and comorbidities. All estimates are negative so that that the level set of the CATE at $\theta = 0$, the reference value for our analysis, is effectively empty.] {\label{fig:cate}
		\centering
		
		\includegraphics[scale=0.3]{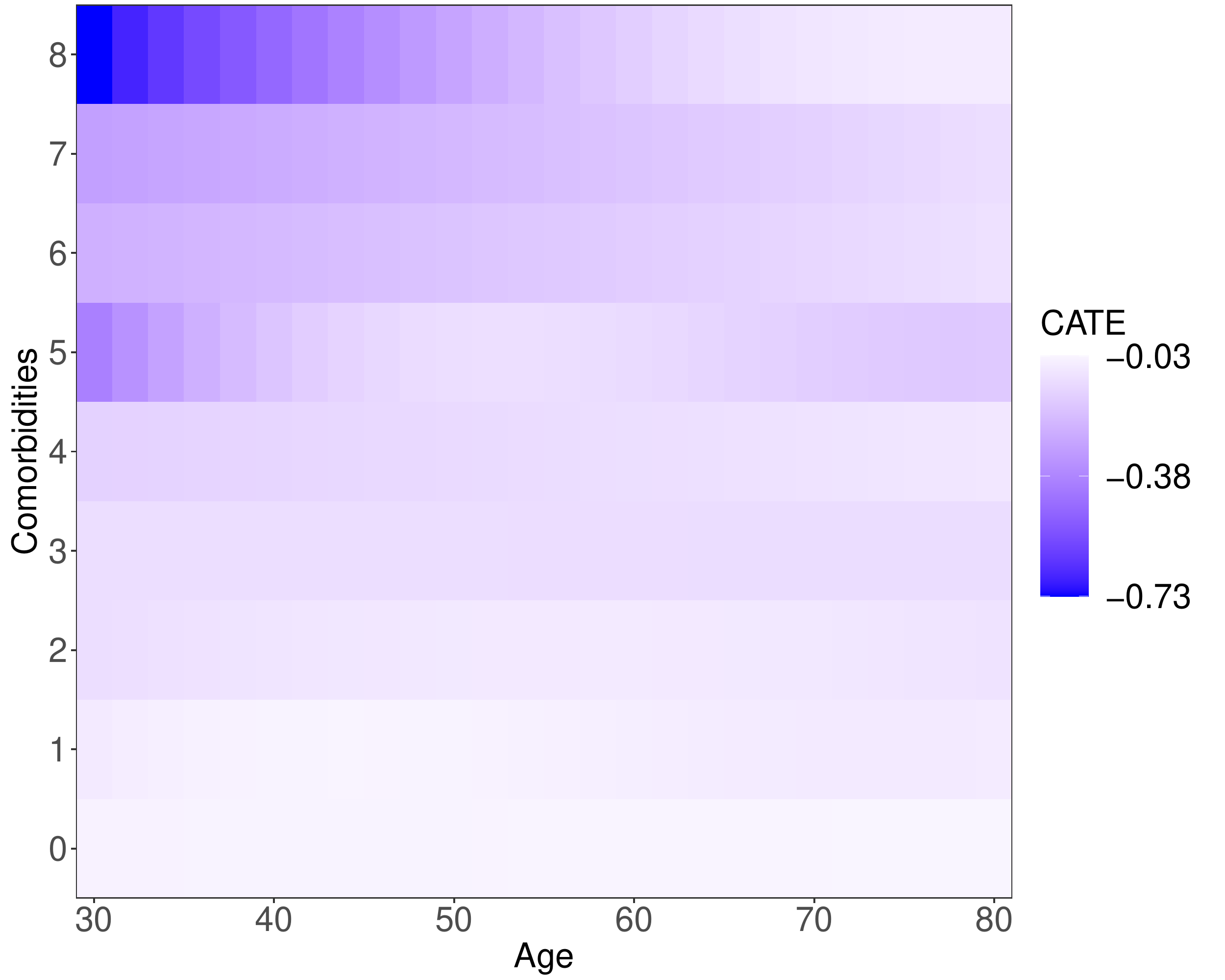}
	}
	\hfill
	\subfloat[The blue regions represent covariates' values where, with high probability, the CATE function is negative. The union of these regions is the complement of $\widehat{C}_u$ defined in Section \ref{section:inference_lvl_sets} at level $\theta = 0$.]{\label{fig:cate_inf}
		\centering
		
		\includegraphics[scale=0.3]{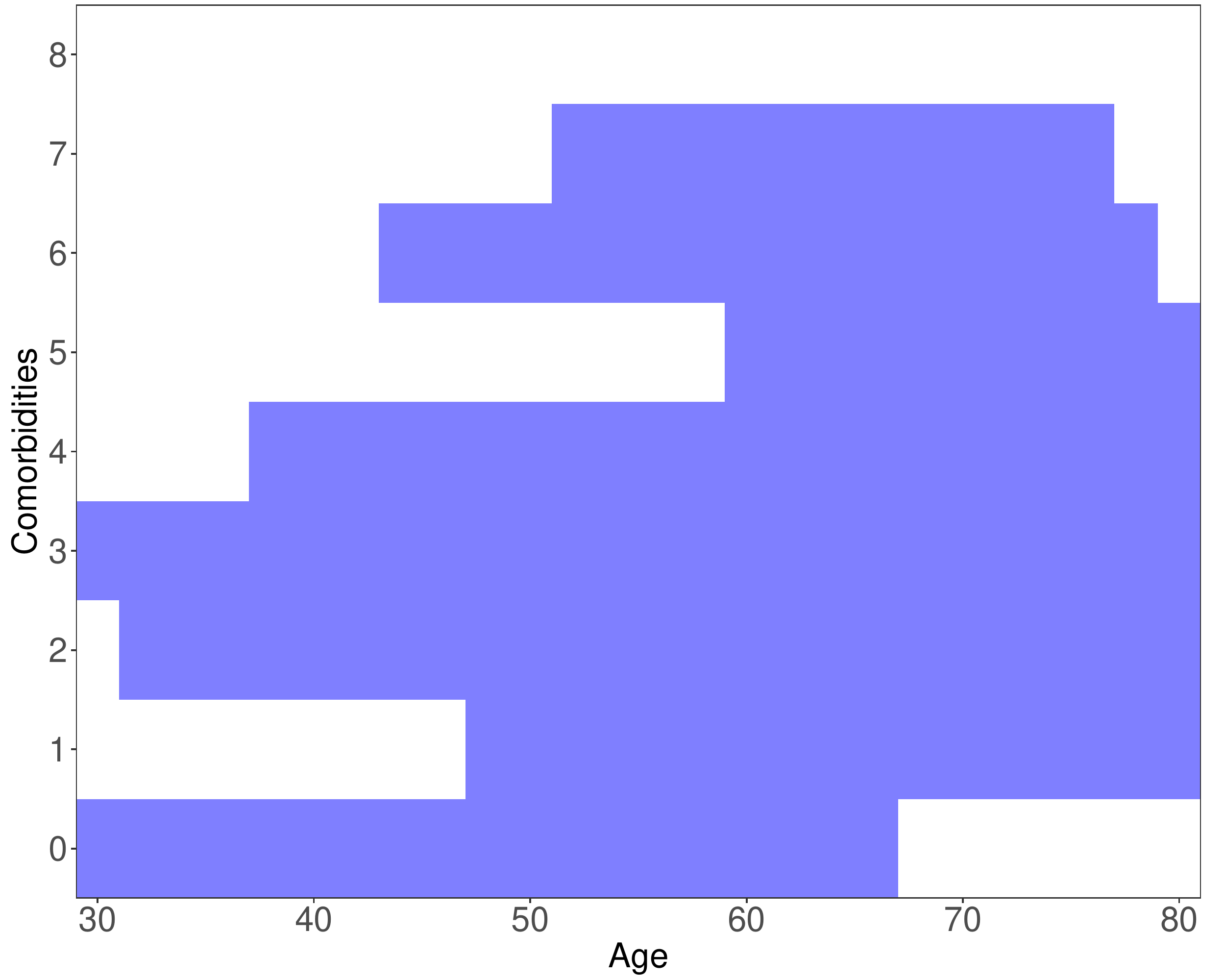}
	}	
	\caption{Data analysis results}
	\label{fig:da}
\end{figure}
\section{Conclusions}
In this work, we have studied the convergence rates for estimating the upper level sets of the conditional average treatment effect (CATE). We have provided upper bounds on the error in estimating this parameter when either DR-Learners or Lp-R-Learners of the CATE are thresholded to yield estimators of the CATE level sets. Furthermore, we have shown that the estimator based on thresholding the Lp-R-Learner is minimax optimal in a particular smoothness model that allows the CATE and the nuisance functions to have different smoothness levels. We have also discussed a straightforward method to construct upper and lower confidence regions for the upper level set.

There are many questions that remain to be investigated. First, implementing the minimax optimal Lp-R-Learner estimator of the CATE presents a few challenges. For example, it would be very useful to study how to choose the right values for the tuning parameters that would adapt to the unknown smoothness of the data generating process. In addition, when the covariates' dimension is large, this estimator requires substantial computational power. Second, our construction used to derive the minimax rates requires that the product of the parameter $\xi$ governing the margin condition times the smoothness $\gamma$ of the CATE is less than the dimension of the covariates. Establishing minimax optimality without imposing this assumption remains an open problem.

It would also be of interest to consider the estimation of related parameters. For example, one could estimate 1) the $\Pb_X$ measure of the CATE upper level set, which could potentially be estimated with even more precision than the upper level set and the CATE itself, 2) the boundary level set at $\theta = 0$, as well as 3) the CATE upper level sets under additional structural constraints, e.g. in cases where the covariates take values on a lower-dim manifold in $\R^d$.

Finally, an important avenue for future work is to consider estimators of CATE upper level sets that are based on empirical risk minimization, as opposed to the one we have considered in this work that consist of simply thresholding estimators of the CATE functions. This would naturally allow the user to pre-specify a family of candidate upper level sets, which can be chosen sufficiently regular, e.g. hyper-rectangles, to have a natural interpretation in the context of the application considered.
\section{Acknowledgments}
MB thanks Profs. Abhishek Ananth, Alejandro Sanchez Becerra, Ruoxuan Xiong, and Miles Lopes, Tudor Manole, the participants at the QTM Seminar at Emory University and those at his doctoral thesis defense for very helpful discussions.

\section{Disclaimer}
The authors declare no conflicts. Research in this article was supported by the National Library of Medicine, \#1R01LM013361-01A1. All statements in this report, including its findings and conclusions, are solely those of the authors. The dataset used for this study was purchased with a grant from the Society of American Gastrointestinal and Endoscopic Surgeons. Although the AMA Physician Masterfile data is the source of the raw physician data, the tables and tabulations were prepared by the authors and do not reflect the work of the AMA. The Pennsylvania Health Cost Containment Council (PHC4) is an independent state agency responsible for addressing the problems of escalating health costs, ensuring the quality of health care, and increasing access to health care for all citizens. While PHC4 has provided data for this study, PHC4 specifically disclaims responsibility for any analyses, interpretations or conclusions. Some of the data used to produce this publication was purchased from or provided by the New York State Department of Health (NYSDOH) Statewide Planning and Research Cooperative System (SPARCS). However, the conclusions derived, and views expressed herein are those of the author(s) and do not reflect the conclusions or views of NYSDOH. NYSDOH, its employees, officers, and agents make no representation, warranty or guarantee as to the accuracy, completeness, currency, or suitability of the information provided here. This publication was derived, in part, from a limited data set supplied by the Florida Agency for Health Care Administration (AHCA) which specifically disclaims responsibility for any analysis, interpretations, or conclusions that may be created as a result of the limited data set.

\bibliographystyle{plainnat}
\bibliography{ref.bib}

\begin{thebibliography}{39}
\providecommand{\natexlab}[1]{#1}
\providecommand{\url}[1]{\texttt{#1}}
\expandafter\ifx\csname urlstyle\endcsname\relax
  \providecommand{\doi}[1]{doi: #1}\else
  \providecommand{\doi}{doi: \begingroup \urlstyle{rm}\Url}\fi

\bibitem[Athey and Imbens(2016)]{athey2016recursive}
Susan Athey and Guido Imbens.
\newblock Recursive partitioning for heterogeneous causal effects.
\newblock \emph{Proceedings of the National Academy of Sciences}, 113\penalty0
  (27):\penalty0 7353--7360, 2016.

\bibitem[Athey and Wager(2021)]{athey2021policy}
Susan Athey and Stefan Wager.
\newblock Policy learning with observational data.
\newblock \emph{Econometrica}, 89\penalty0 (1):\penalty0 133--161, 2021.

\bibitem[Audibert and Tsybakov(2007)]{audibert2007fast}
Jean-Yves Audibert and Alexandre~B Tsybakov.
\newblock Fast learning rates for plug-in classifiers.
\newblock \emph{The Annals of statistics}, 35\penalty0 (2):\penalty0 608--633,
  2007.

\bibitem[Ben-Michael et~al.(2022)Ben-Michael, Imai, and Jiang]{ben2022policy}
Eli Ben-Michael, Kosuke Imai, and Zhichao Jiang.
\newblock Policy learning with asymmetric utilities.
\newblock \emph{arXiv preprint arXiv:2206.10479}, 2022.

\bibitem[Cattaneo et~al.(2022)Cattaneo, Chandak, Jansson, and
  Ma]{cattaneo2022boundary}
Matias~D Cattaneo, Rajita Chandak, Michael Jansson, and Xinwei Ma.
\newblock Boundary adaptive local polynomial conditional density estimators.
\newblock \emph{arXiv preprint arXiv:2204.10359}, 2022.

\bibitem[Chakraborty and Moodie(2013)]{chakraborty2013statistical}
Bibhas Chakraborty and Erica~E Moodie.
\newblock Statistical methods for dynamic treatment regimes.
\newblock \emph{Springer-Verlag. doi}, 10:\penalty0 978--1, 2013.

\bibitem[Chen et~al.(2017)Chen, Genovese, and Wasserman]{chen2017density}
Yen-Chi Chen, Christopher~R Genovese, and Larry Wasserman.
\newblock Density level sets: Asymptotics, inference, and visualization.
\newblock \emph{Journal of the American Statistical Association}, 112\penalty0
  (520):\penalty0 1684--1696, 2017.

\bibitem[de~la Pe{\~n}a and Montgomery-Smith(1995)]{de1995decoupling}
Victor~H de~la Pe{\~n}a and Stephen~J Montgomery-Smith.
\newblock Decoupling inequalities for the tail probabilities of multivariate
  u-statistics.
\newblock \emph{The Annals of Probability}, pages 806--816, 1995.

\bibitem[Elixhauser et~al.(1998)Elixhauser, Steiner, Harris, and
  Coffey]{elixhauser1998comorbidity}
Anne Elixhauser, Claudia Steiner, D~Robert Harris, and Rosanna~M Coffey.
\newblock Comorbidity measures for use with administrative data.
\newblock \emph{Medical care}, 36\penalty0 (1):\penalty0 8--27, 1998.

\bibitem[Foster and Syrgkanis(2019)]{foster2019orthogonal}
Dylan~J Foster and Vasilis Syrgkanis.
\newblock Orthogonal statistical learning.
\newblock \emph{arXiv preprint arXiv:1901.09036}, 2019.

\bibitem[Gin{\'e} et~al.(2000)Gin{\'e}, Lata{\l}a, and
  Zinn]{gine2000exponential}
Evarist Gin{\'e}, Rafa{\l} Lata{\l}a, and Joel Zinn.
\newblock Exponential and moment inequalities for u-statistics.
\newblock In \emph{High Dimensional Probability II}, pages 13--38. Springer,
  2000.

\bibitem[Gur et~al.(2022)Gur, Momeni, and Wager]{gur2022smoothness}
Yonatan Gur, Ahmadreza Momeni, and Stefan Wager.
\newblock Smoothness-adaptive contextual bandits.
\newblock \emph{Operations Research}, 2022.

\bibitem[Hahn et~al.(2020)Hahn, Murray, and Carvalho]{hahn2020bayesian}
P~Richard Hahn, Jared~S Murray, and Carlos~M Carvalho.
\newblock Bayesian regression tree models for causal inference: Regularization,
  confounding, and heterogeneous effects (with discussion).
\newblock \emph{Bayesian Analysis}, 15\penalty0 (3):\penalty0 965--1056, 2020.

\bibitem[Hirano and Porter(2009)]{hirano2009asymptotics}
Keisuke Hirano and Jack~R Porter.
\newblock Asymptotics for statistical treatment rules.
\newblock \emph{Econometrica}, 77\penalty0 (5):\penalty0 1683--1701, 2009.

\bibitem[Imai and Ratkovic(2013)]{imai2013estimating}
Kosuke Imai and Marc Ratkovic.
\newblock Estimating treatment effect heterogeneity in randomized program
  evaluation.
\newblock \emph{The Annals of Applied Statistics}, 7\penalty0 (1):\penalty0
  443--470, 2013.

\bibitem[Kannan et~al.(2015)Kannan, Reddy, Mukerji, Parithivel, Shah,
  Gilchrist, and Farkas]{kannan2015laparoscopic}
Umashankkar Kannan, Vemuru Sunil~K Reddy, Amar~N Mukerji, Vellore~S Parithivel,
  Ajay~K Shah, Brian~F Gilchrist, and Daniel~T Farkas.
\newblock Laparoscopic vs open partial colectomy in elderly patients: Insights
  from the american college of surgeons-national surgical quality improvement
  program database.
\newblock \emph{World Journal of Gastroenterology}, 21\penalty0 (45):\penalty0
  12843, 2015.

\bibitem[Kemp and Finlayson(2008)]{kemp2008outcomes}
Jason~A Kemp and Samuel~RG Finlayson.
\newblock Outcomes of laparoscopic and open colectomy: a national
  population-based comparison.
\newblock \emph{Surgical innovation}, 15\penalty0 (4):\penalty0 277--283, 2008.

\bibitem[Kennedy(2020)]{kennedy2020optimal}
Edward~H Kennedy.
\newblock Optimal doubly robust estimation of heterogeneous causal effects.
\newblock \emph{arXiv preprint arXiv:2004.14497}, 2020.

\bibitem[Kennedy et~al.(2020)Kennedy, Balakrishnan, and
  G’Sell]{kennedy2020sharp}
Edward~H Kennedy, Sivaraman Balakrishnan, and Max G’Sell.
\newblock Sharp instruments for classifying compliers and generalizing causal
  effects.
\newblock \emph{The Annals of Statistics}, 48\penalty0 (4):\penalty0
  2008--2030, 2020.

\bibitem[Kennedy et~al.(2022)Kennedy, Balakrishnan, and
  Wasserman]{kennedy2022minimax}
Edward~H Kennedy, Sivaraman Balakrishnan, and Larry Wasserman.
\newblock Minimax rates for heterogeneous effect estimation.
\newblock \emph{arXiv preprint arXiv:}, 2022.

\bibitem[K{\"u}nzel et~al.(2019)K{\"u}nzel, Sekhon, Bickel, and
  Yu]{kunzel2019metalearners}
S{\"o}ren~R K{\"u}nzel, Jasjeet~S Sekhon, Peter~J Bickel, and Bin Yu.
\newblock Metalearners for estimating heterogeneous treatment effects using
  machine learning.
\newblock \emph{Proceedings of the national academy of sciences}, 116\penalty0
  (10):\penalty0 4156--4165, 2019.

\bibitem[Luedtke and Van Der~Laan(2016)]{luedtke2016statistical}
Alexander~R Luedtke and Mark~J Van Der~Laan.
\newblock Statistical inference for the mean outcome under a possibly
  non-unique optimal treatment strategy.
\newblock \emph{Annals of statistics}, 44\penalty0 (2):\penalty0 713, 2016.

\bibitem[Mammen and Polonik(2013)]{mammen2013confidence}
Enno Mammen and Wolfgang Polonik.
\newblock Confidence regions for level sets.
\newblock \emph{Journal of Multivariate Analysis}, 122:\penalty0 202--214,
  2013.

\bibitem[Nie and Wager(2021)]{nie2021quasi}
Xinkun Nie and Stefan Wager.
\newblock Quasi-oracle estimation of heterogeneous treatment effects.
\newblock \emph{Biometrika}, 108\penalty0 (2):\penalty0 299--319, 2021.

\bibitem[Qiao and Polonik(2019)]{qiao2019nonparametric}
Wanli Qiao and Wolfgang Polonik.
\newblock Nonparametric confidence regions for level sets: Statistical
  properties and geometry.
\newblock \emph{Electronic Journal of Statistics}, 13\penalty0 (1):\penalty0
  985--1030, 2019.

\bibitem[Reeve et~al.(2021)Reeve, Cannings, and Samworth]{reeve2021optimal}
Henry~WJ Reeve, Timothy~I Cannings, and Richard~J Samworth.
\newblock Optimal subgroup selection.
\newblock \emph{arXiv preprint arXiv:2109.01077}, 2021.

\bibitem[Rigollet and Vert(2009)]{rigollet2009optimal}
Philippe Rigollet and R{\'e}gis Vert.
\newblock Optimal rates for plug-in estimators of density level sets.
\newblock \emph{Bernoulli}, 15\penalty0 (4):\penalty0 1154--1178, 2009.

\bibitem[Robins et~al.(2009)Robins, Tchetgen, Li, and van~der
  Vaart]{robins2009semiparametric}
James Robins, Eric~Tchetgen Tchetgen, Lingling Li, and Aad van~der Vaart.
\newblock Semiparametric minimax rates.
\newblock \emph{Electronic journal of statistics}, 3:\penalty0 1305, 2009.

\bibitem[Robins(2004)]{robins2004optimal}
James~M Robins.
\newblock Optimal structural nested models for optimal sequential decisions.
\newblock In \emph{Proceedings of the Second Seattle Symposium in
  Biostatistics: analysis of correlated data}, pages 189--326. Springer, 2004.

\bibitem[Robins et~al.(2017)Robins, Li, Mukherjee, Tchetgen, and van~der
  Vaart]{robins2017minimax}
James~M Robins, Lingling Li, Rajarshi Mukherjee, Eric~Tchetgen Tchetgen, and
  Aad van~der Vaart.
\newblock Minimax estimation of a functional on a structured high-dimensional
  model.
\newblock \emph{The Annals of Statistics}, 45\penalty0 (5):\penalty0
  1951--1987, 2017.

\bibitem[Semenova and Chernozhukov(2021)]{semenova2021debiased}
Vira Semenova and Victor Chernozhukov.
\newblock Debiased machine learning of conditional average treatment effects
  and other causal functions.
\newblock \emph{The Econometrics Journal}, 24\penalty0 (2):\penalty0 264--289,
  2021.

\bibitem[Shalit et~al.(2017)Shalit, Johansson, and
  Sontag]{shalit2017estimating}
Uri Shalit, Fredrik~D Johansson, and David Sontag.
\newblock Estimating individual treatment effect: generalization bounds and
  algorithms.
\newblock In \emph{International Conference on Machine Learning}, pages
  3076--3085. PMLR, 2017.

\bibitem[Takatsu and Westling(2022)]{takatsu2022debiased}
Kenta Takatsu and Ted Westling.
\newblock Debiased inference for a covariate-adjusted regression function.
\newblock \emph{arXiv preprint arXiv:2210.06448}, 2022.

\bibitem[Tsybakov(2009)]{tsybakov2004introduction}
Alexandre~B Tsybakov.
\newblock \emph{Introduction to nonparametric estimation}.
\newblock Springer Series in Statistics. Springer, New York, NY, 2009.

\bibitem[Varela et~al.(2008)Varela, Asolati, Huerta, and
  Anthony]{varela2008outcomes}
J~Esteban Varela, Massimo Asolati, Sergio Huerta, and Thomas Anthony.
\newblock Outcomes of laparoscopic and open colectomy at academic centers.
\newblock \emph{The American Journal of Surgery}, 196\penalty0 (3):\penalty0
  403--406, 2008.

\bibitem[Wager and Athey(2018)]{wager2018estimation}
Stefan Wager and Susan Athey.
\newblock Estimation and inference of heterogeneous treatment effects using
  random forests.
\newblock \emph{Journal of the American Statistical Association}, 113\penalty0
  (523):\penalty0 1228--1242, 2018.

\bibitem[Willett and Nowak(2007)]{willett2007minimax}
Rebecca~M Willett and Robert~D Nowak.
\newblock Minimax optimal level-set estimation.
\newblock \emph{IEEE Transactions on Image Processing}, 16\penalty0
  (12):\penalty0 2965--2979, 2007.

\bibitem[Wu et~al.(2022)Wu, Li, Tu, Zheng, and Chen]{wu2022comparison}
Jini Wu, Bo~Li, Shiliang Tu, Boan Zheng, and Bingchen Chen.
\newblock Comparison of laparoscopic and open colectomy for splenic flexure
  colon cancer: a systematic review and meta-analysis.
\newblock \emph{International Journal of Colorectal Disease}, pages 1--11,
  2022.

\bibitem[Wu et~al.(2010)Wu, He, Zhou, Ke, and Lan]{wu2010role}
Xiao-Jian Wu, Xiao-Sheng He, Xu-Yu Zhou, Jia Ke, and Ping Lan.
\newblock The role of laparoscopic surgery for ulcerative colitis: systematic
  review with meta-analysis.
\newblock \emph{International journal of colorectal disease}, 25\penalty0
  (8):\penalty0 949--957, 2010.

\end{thebibliography}
\appendix
\section{Proof of Lemma \ref{lemma:lemma3.1}}
Write
\begin{align*}
	\E\{d_H(\widehat\Gamma, \Gamma)\} = \E\int_{\widehat\Gamma^c \cap \Gamma} |\tau(x) - \theta| f(x) dx + \E\int_{\widehat\Gamma \cap \Gamma^c} |\tau(x) - \theta| f(x) dx
\end{align*}
Let $\alpha_n = c_\alpha a_n$ and $\beta_n = c_\beta (\log n)^{-\min\{(1+\xi)^{-1}, \epsilon\}}$, where $c_\alpha = 2c_a$ and 
$$c_\beta = \max[c_b, c_1\{\mu(1 + \xi) / c_7\}^{1/\kappa_2}].$$ Consider the first term and write
\begin{align*}
	& \widehat\Gamma^c \cap \Gamma = \{x \in \mathcal{X}: \widehat\tau(x) \leq \theta \text{ and } \tau(x) > \theta\} = A_1 \cup A_2 \cup A_3 \\
	& A_1 =  \{x \in \mathcal{X}: \widehat\tau(x) \leq \theta  \text{ and } \theta < \tau(x) \leq \theta + \alpha_n\} \\
	& A_2 =  \{x \in \mathcal{X}: \widehat\tau(x) \leq \theta  \text{ and } \theta + \alpha_n < \tau(x) \leq \theta + \beta_n\} \\
	& A_3 =  \{x \in \mathcal{X}: \widehat\tau(x) \leq \theta  \text{ and } \tau(x) > \theta + \beta_n\}
\end{align*}
Let $n_0$ be the integer such that for all $n \geq n_0$, it holds that
\begin{align*}
	\alpha_n < \beta_n < \min(\eta, t_0, \Delta)
\end{align*}
For the proof we assume that the sample size $n$ exceeds $n_0$. 

We have $A_1 \subseteq \{x \in \mathcal{X}: 0< |\tau(x) - \theta| \leq \alpha_n\}$. By Assumption \ref{assumption_margin} (margin condition), we have
\begin{align*}
	\E\int_{A_1} |\tau(x) - \theta| f(x) dx \leq \alpha_n \Pb(0 < |\tau(X) - \theta| \leq a_n) \leq c_0 \alpha_n^{1 + \xi}
\end{align*}
Next, let $J_n = \lfloor \log_2 \{\beta_n / \alpha_n\} \rfloor + 1$. Notice that $\beta_n / \alpha_n \lesssim n^{\mu}$ so that $J_n \lesssim \log n$.  Partition $A_2$ as $A_2= \cup_{j = 1}^{J_n} A_2 \cap \mathcal{V}_j$, where
\begin{align*}
	\mathcal{V}_j = \{x \in \mathcal{X}: \widehat\tau(x) \leq \theta \text{ and } 2^{j - 1} \alpha_n < \tau(x) - \theta \leq 2^{j}\alpha_n\}
\end{align*}
We have
\begin{align*}
	\E\int_{A_2} |\tau(x) - \theta| f(x) dx = \sum_{i = 1}^{J_n} \E\int_{A_2 \cap \mathcal{V}_j} |\tau(x) - \theta| f(x) dx
\end{align*}
and
\begin{align*}
	\mathcal{V}_j \subset \{x \in \mathcal{X}: |\widehat\tau(x) - \tau(x)| > 2^{j - 1} \alpha_n \text{ and } |\tau(x) - \theta| < 2^{j}\alpha_n\} \cap  \mathcal{D}(\min(\eta, t_0))
\end{align*}
To see why this is the case, consider a $x^*$ such that $\tau(x^*) \in (\theta + 2^{j-1}\alpha_n, \theta + 2^j\alpha_n]$ and $\widehat\tau(x^*) \leq \theta$. Clearly, $x^*$ satisfies $\tau(x^*) \in [\theta - 2^{j}a_n, \theta + 2^j\alpha_n]$. Notice that 
$$\tau(x^*) - 2^{j-1}\alpha_n > \theta + 2^{j-1}\alpha_n - 2^{j-1}\alpha_n = \theta \geq \widehat\tau(x^*)$$
for any $j$. The claim follows because we have shown that $\widehat\tau(x^*) < \tau(x^*) - 2^{j-1}\alpha_n$ and thus $x^*$ is in the larger set.

For any $j \geq 1$, we have that $2^{j-1} \alpha_n > c_a a_n$ and
\begin{align*}
	\E\int _{A_2 \cap \mathcal{V}_j} |\tau(x) - \theta| f(x) dx & \leq \|f\|_\infty 2^{j}\alpha_n\int_{\mathcal{X}} \Pb\left(\left|\widehat\tau(x) - \tau(x)\right| > 2^{j-1}\alpha_n \right) \one\{0 < |\tau(x) - \theta| < 2^{j}\alpha_n\} dx\\
	& \leq \| f\|_\infty c_02^{j(1 + \xi)}\alpha_n^{1+\xi} \left\{c_3\exp(-c_42^{(j-1)\kappa_1} c_\alpha^{\kappa_1}) + c_5\frac{\delta^{1 + \xi}_n}{2^{(j - 1)(1 + \xi)} \alpha_n^{1+\xi}} \right\} \\
	& = \| f\|_\infty c_3c_02^{j(1 + \xi)}\alpha_n^{1+\xi} \exp(-c_42^{(j-1)\kappa_1} c_\alpha^{\kappa_1}) + c_5c_02^{1 + \xi} \delta_n^{1 + \xi}
\end{align*}
Thus,
\begin{align*}
	\E\int _{A_2} |\tau(x) - \theta| f(x) dx & = \sum_{j = 1}^{J_n} \E\int _{A_2 \cap \mathcal{V}_j} |\tau(x) - \theta| f(x) dx  \\
	& \leq \| f\|_\infty c_3c_0\alpha_n^{1+\xi}  \sum_{j = 1}^{J_n} 2^{j(1 + \xi)}\exp\left\{-c_4 \left(\frac{c_\alpha}{2}\right)^{\kappa_1}2^{ j\kappa_1}\right\} + J_n c_5c_0 2^{1 + \xi}\delta_n^{1 + \xi} \\
	& \lesssim a_n^{1+\xi} + \delta_n^{1 + \xi} \log n
\end{align*}
The last inequality follows, because for any $a, b, c > 0$, $\sum_{j = 1}^\infty 2^{aj}\exp(-b2^{jc}) < \infty$. In fact, for any $\alpha$, there exists a constant $C$ such that $(1/e^b)^{x} \leq Cx^{-\alpha}$ for any $x \geq 1$. Let $j_0$ be large enough so that $2^{jc - 1} \geq j \geq 1$ for all $j \geq j_0$.  Then, for some constant $C$ and $x_j = 2^{jc- 1}$:
\begin{align*}
	2^{aj}\left(\frac{1}{e^b}\right)^{2^{jc - 1}}= 2^{a/c} x_j^{a/c}\left(\frac{1}{e^b}\right)^{x_j} \leq C \text{ for all } j \geq j_0.
\end{align*}
Then, we have
\begin{align*}
	\sum_{j = 1}^\infty 2^{aj}\exp(-b2^{jc}) & = \sum_{j = 1}^{j_0 - 1} 2^{aj}\exp(-b2^{jc}) + \sum_{j = j_0}^\infty 2^{aj}\exp(-b2^{jc - 1}) \exp(-b2^{jc - 1}) \\
	& \leq \sum_{j = 1}^{j_0 - 1} 2^{aj}\exp(-b2^{jc}) + C\sum_{j = j_0}^\infty \exp(-b2^{jc - 1}) \\
	& < \infty
\end{align*}
Finally, we have that $\beta_n > c_b b_n$ so that
\begin{align*}
	& \E \int_{A_3}|\tau(x) - \theta| f(x) dx  \leq \int_{\mathcal{X}} |\tau(x) - \theta| \Pb\left( |\widehat\tau(x)-\tau(x)|> \beta_n \right) f(x) dx \\
	& \leq \int_{\mathcal{X}} |\tau(x) - \theta| f(x) dx\left(c_6\exp\left[-c_7\left(\frac{c_\beta}{c_1}\right)^{\kappa_2} \left\{\frac{(\log n)^{1/\kappa_2 + \epsilon}}{(\log n)^{\min\{(1+\xi)^{-1}, \epsilon\}}} \right\}^{\kappa_2}\right] + \frac{c_8}{c_\beta^{1+\xi}} \delta_n^{1+\xi} \log n\right) \\
	& \lesssim \exp\left\{-c_7\left(\frac{c_\beta}{c_1}\right)^{\kappa_2} \log n \right\} + c_8\delta_n^{1+\xi} \log n \\
	& = \exp\left[-\max\{c_7(c_b / c_1)^{\kappa_2}, \mu(1 + \xi)\} \log n \right] + c_8 \delta_n^{1+\xi} \log n \\
	& \leq n^{-(1+\xi)\mu} + c_8 \delta_n^{1+\xi} \log n \\
	& \lesssim a_n^{1+\xi} + c_8\delta_n^{1+\xi} \log n
\end{align*}
The bound on $ \E\int_{\widehat\Gamma \cap \Gamma^c} |\tau(x) - \theta| f(x) dx$ follows similarly.

\section{Proof of Lemma  \ref{lemma:moment_bound_dr} \label{app:proof_moment_bound_dr}}
By definition, we have $\widehat\tau(x) = \sum_{i = 1}^n W_i(x; X^n) \widehat\varphi(Z_i)$. Define $\overline\tau(x; X^n) = \sum_{i = 1}^n W_i(x; X^n)\tau(X_i)$ and recall that $\Delta(x; X^n) = \overline\tau(x; X^n) - \tau(x)$. Finally, let $\widehat{b}(X_i) = \E\{\widehat\varphi(Z_i) - \varphi(Z_i) \mid X_i, D^n\}$. We start from the decomposition
\begin{align*}
	\widehat\tau(x) - \tau(x) & =  \sum_{i = 1}^n W_i(x; X^n) \varphi(Z_i) - \overline\tau(x; X^n)  + \sum_{i = 1}^n W_i(x; X^n)\{ \widehat\varphi(Z_i) - \varphi(Z_i) - \widehat{b}(X_i)\} \\
	& \hphantom{=} + \Delta(x; X^n) + \sum_{i = 1}^n W_i(x; X^n) \widehat{b}(X_i).
\end{align*}
Given $(D^n, X^n)$, the last two terms are constants, whereas the first two are mean-zero. We have
\begin{align*}
	&	\Pb\left(\left| \widehat\tau(x) - \tau(x) \right| > t \right) \\
	& \leq \E\left\{\Pb\left(\left|\sum_{i = 1}^n W_i(x; X^n) \varphi(Z_i) - \overline\tau(x; X^n) \right| > \frac{t}{3} - \frac{|\Delta(x; X_n)|}{2} \mid X^n \right) \right\} \\
	& \hphantom{\leq} + \E\left[\Pb\left(\left| \sum_{i = 1}^n W_i(x; X^n)\{ \widehat\varphi(Z_i) - \varphi(Z_i) - \widehat{b}(X_i)\} \right| > \frac{t}{3} - \frac{|\Delta(x; X_n)|}{2} \mid D^n, X^n \right)\right] \\
	& \hphantom{\leq} + \inf_{p > 0} \left(\frac{3}{t}\right)^p \E \left| \sum_{i = 1}^n W_i(x; X^n) \widehat{b}(X_i) \right| ^p
\end{align*}
As in equation 2.16 in \cite{gine2000exponential}, we have, for $p \geq 2$:
\begin{align*}
	& \E\left[\left|\sum_{i = 1}^n W_i(x; X^n) \{\varphi(Z_i) - \tau(X_i)\}\right|^p \mid X^n \right] \\
	& \leq 2^p(p-1)^{p/2} \E\left(\left[\sum_{i = 1}^n W^2_i(x; X^n) \{\varphi(Z_i) - \tau(X_i)\}^2\right]^{p/2} \mid X^n\right) \\
	& \leq (4\|\varphi\|_\infty)^pp^{p/2} S^p(x; X^n)
\end{align*}
Thus,
\begin{align*}
	\E\left[\left|\sum_{i = 1}^n W_i(x; X^n) \varphi(Z_i) - \overline\tau(x; X^n)\right|^p \right] & \leq (4\|\varphi\|_\infty)^pp^{p/2} \E\{S^p(x; X^n)\} \\
	& \equiv (4\|\varphi\|_\infty)^pp^{p/2} s_n^p 
\end{align*}
The following lemma, which can be found for instance in \cite{gine2000exponential} (eq. 3.2, page 14), shows that an exponential inequality follows if all the moments $\E| \widehat\tau(x) - \tau(x)|^p$ are properly controlled. 
\begin{lemma}\label{lemma:from_moment_to_exp}
	Let $X$ be some random variable such that $\E|X|^p \leq a_n^p p^{p / \alpha}$, for all $p \geq p_0$ and some fixed $p_0$. Then,
	\begin{align*}
		\Pb\left(|X| > t \right) \leq e^{p_0} \exp\left\{-\left(\frac{t}{a_ne}\right)^\alpha\right\}
	\end{align*}
\end{lemma}
\begin{proof}
	For $t$ such that $p = (te^{-1}a_n^{-1})^{\alpha} \geq p_0$, we have the bound $t^{-p}\E|X|^p \leq e^{-p}$. For all $t$, we thus have $\Pb(|X| > t) \leq e^{p_0 - p}$, since for values of $t$ for which $p < p_0$,  $e^{p_0 - p} > 1$ is a valid bound. 
\end{proof}
In light of Lemma \ref{lemma:from_moment_to_exp}, we have for all $t \geq 3|\Delta(x; X^n)|$:
\begin{align*}
	& \Pb\left(\left|\sum_{i = 1}^n W_i(x; X^n) \varphi(Z_i) - \overline\tau(x; X^n) \right| > \frac{t}{3} - \frac{|\Delta(x; X_n)|}{2} \mid X^n \right) \\
	&  \leq \Pb\left(\left|\sum_{i = 1}^n W_i(x; X^n) \varphi(Z_i) - \overline\tau(x; X^n) \right| > \frac{t}{6} \mid X^n \right) \\
	& \leq e^2\exp\left\{ -  \left(\frac{t}{24e \|\varphi\|_\infty s_n }\right)^2 \right\}
\end{align*}
Similarly,
\begin{align*}
	& \Pb\left(\left|\sum_{i = 1}^n W_i(x; X^n) \{ \widehat\varphi(Z_i) - \varphi(Z_i) - \widehat{b}(X_i)\}  \right| > \frac{t}{3} - \frac{|\Delta(x; X_n)|}{2} \mid D^n, X^n \right) \\
	& \leq e^2\exp\left\{ -  \left(\frac{t}{12c_2\|\varphi\|_\infty e s_n }\right)^2 \right\}
\end{align*}
under the assumption that $\|\widehat\varphi - \varphi - \widehat{b}\|_\infty \leq c_2\|\varphi\|_\infty$. Therefore, we conclude that
\begin{align*}
	\Pb\left(\left| \widehat\tau(x) - \tau(x) \right| > t \right) \leq  2e^2 \exp\left\{-\left(\frac{t}{12(c_2 \lor 2) e \|\varphi\|_\infty s_n} \right)^2\right\} + 3^{1 + \xi}\left(\frac{\delta_n}{t}\right)^{1 + \xi}
\end{align*}
for all $t \geq 3|\Delta(x; X^n)|$.

\section{Proof of Lemma \ref{lemma:moment_lp-r-learner}}  
Recall that the Lp-R-Learner can be written as $\widehat\tau(x_0) = \rho_h(x_0)^T \widehat{Q}^{-1} \widehat{R}$, where
\begin{align*}
	\widehat{Q} = \mathbb{U}_n \{\widehat{f}_1(Z_1, Z_2)\} \quad \text{and} \quad \widehat{R} = \mathbb{U}_n\{\widehat{f}_2(Z_1, Z_2)\}
\end{align*}
for some functions $\widehat{f}_1$ and $\widehat{f}_2$ described in Definition \ref{def:lp-r-learner}. Define $\tau_h(x_0) = \rho_h(x_0)^TQ^{-1}R$ to be the projection parameter. By proposition 4 in \cite{kennedy2022minimax}, it holds that
\begin{align*}
	\left|\tau_h(x_0) - \tau(x_0) \right| \leq \begin{cases} 
		& c_1h^\gamma \quad \text{ if } x_0 \in D(\eta) \\
		& c_2h^{\gamma'} \quad \text{ if } x_0 \not\in D(\eta)
	\end{cases}
\end{align*}
for some constants $c_1$ and $c_2$. Let $\widehat{S} = \widehat{R} - \widehat{Q}Q^{-1}R$. By Proposition 6 in \cite{kennedy2022minimax}, we have, under the conditions of the theorem, for $J = {{d + \lfloor \gamma \rfloor} \choose {\lfloor \gamma \rfloor}}$ and a constant $c_3$:
\begin{align*}
	\{\widehat\tau(x_0) - \tau_h(x_0)\}^2 \leq c_3\sum_{j=1}^J \widehat{S}^2_j
\end{align*}
Therefore, if $x_0 \in D(\eta)$ and $t \geq 2c_1h^\gamma$:
\begin{align*}
	&	\Pb\left( \left|\widehat\tau(x_0) - \tau(x_0)\right| > t \right) \leq 	\Pb\left( \left|\widehat\tau(x_0) - \tau_h(x_0)\right| > t - c_1h^\gamma \right) \\
	& \leq 	\Pb\left( \left\{\widehat\tau(x_0) - \tau_h(x_0)\right\}^2> \frac{t^2}{4} \right) \\
	& \leq \sum_{j = 1}^J \Pb\left( \left| \widehat{S}_j \right| > \frac{t}{2\sqrt{c_3J}} \right)
\end{align*}
Next, we bound
\begin{align*}
	\Pb\left( \left| \widehat{S}_j \right| > \frac{t}{2\sqrt{c_3J}} \right)
\end{align*}
Recall that $\widehat{S}_j = \mathbb{U}_n \{g(Z_1, Z_2)\}$, where
\begin{align*}
	&g(Z_1, Z_2) = f_{2j}(Z_1, Z_2) - [\widehat{Q}Q^{-1}R]_j \\
	& f_{2j} (Z_1, Z_2) = \rho_{hj}(X_1) K_h(X_1)\widehat\varphi_{y1}(Z_1)+ \rho_{hj}(X_1) K_{h}(X_1)\widehat\varphi_{y2}(Z_1, Z_2) K_{h}(X_2) \\
	& \varphi_{y1}(Z_1) = \{Y -\mu_0(X_1)\}\{A - \pi(X_1)\} \\
	& \varphi_{y2}(Z_1, Z_2) = -\{A_1 - \pi(X_1)\} b^T_{h}(X_1) \Omega^{-1} b_h(X_2)\{Y_2 - \mu_0(X_2)\} \\
	& \Omega = \E\{b_h(X)b_h(X)^T\}, \text{ for } K_h(x) = h^{-d}\one(2\|x - x_0\|\leq h).
\end{align*}
It will be useful to write $\widehat{S}_j$ as a sum of degenerate $U$-statistics, as follows:
\begin{align*}
	\widehat{S}_j  & =	\mathbb{U}_n g(Z_1, Z_2) = \mathbb{U}_n\{g_D(Z_1, Z_2)\} + \Pn\{ g_1(Z_1) \} + \Pn\{g_2(Z_2)\} + \int g(z_1, z_2) d\Pb(z_1)d\Pb(z_2)
\end{align*}
where
\begin{align*}
	& g_D(Z_1, Z_2) = g(Z_1, Z_2) - \int g(z_1, Z_2)d\Pb(z_1) - \int g(Z_1, z_2) d\Pb(z_2) + \int g(z_1, z_2) d\Pb(z_1)d\Pb(z_2) \\
	& g_1(Z_1) = \int g(Z_1, z_2) d\Pb(z_2) - \int g(z_1, z_2) d\Pb(z_1)d\Pb(z_2) \\
	& g_2(Z_2) = \int g(z_1, Z_2) d\Pb(z_1) - \int g(z_1, z_2) d\Pb(z_1)d\Pb(z_2)
\end{align*}
Thus, we have
\begin{align*}
	&\Pb\left(\left|\widehat{S}_j \right| > \frac{t}{2\sqrt{c_3J}}\right) \\
	& \leq
	\Pb\left(\left| \mathbb{U}_n\{g_D(Z_1, Z_2)\}  \right| + \left|\Pn\{ g_1(Z_1) \}\right| +\left| \Pn\{g_2(Z_2)\} \right|  > \frac{t}{2\sqrt{c_3J}} - \left|\int g(z_1, z_2) d\Pb(z_1)d\Pb(z_2)\right|\right)
\end{align*}
By proposition 9 in \cite{kennedy2022minimax}, there exists a constant $c_4$ such that 
\begin{align*}
	\left| \int g(z_1, z_2)d\Pb(z_1)d\Pb(z_2)\right|  \leq c_4 \left(\frac{k}{h^d} \right)^{-2s/d}
\end{align*} 
Therefore, for $t \geq 4\sqrt{c_3J} c_4 \left(\frac{k}{h^d} \right)^{-2s/d}$:
\begin{align*}
	& \Pb\left(\left|\widehat{S}_j - S_j \right| > \frac{t}{2\sqrt{2c_3J}}\right)  \leq \Pb\left(\left| \mathbb{U}_n\{g_D(Z_1, Z_2)\}  \right| + \left|\Pn\{ g_1(Z_1) \}\right| +\left| \Pn\{g_2(Z_2)\} \right|  > \frac{t}{4\sqrt{c_3J}} \right) \\
	& \leq \Pb\left(\left| \mathbb{U}_n\{g_D(Z_1, Z_2)\}  \right| > \frac{t}{12 \sqrt{c_3J}} \right) + \Pb\left(\left|\Pn\{ g_1(Z_1) \}\right| > \frac{t}{12\sqrt{c_3J}} \right) \\
	& \hphantom{\leq} +\Pb\left(\left| \Pn\{g_2(Z_2)\} \right|  > \frac{t}{12\sqrt{c_3J}} \right)
\end{align*}
The second and third terms can be analyzed by Bernstein's inequality. For the first term, we use a concentration inequality for $U$-statistics derived in \cite{gine2000exponential} and restated below. See also \cite{cattaneo2022boundary} for a similar use of this lemma.
\begin{lemma}[Equation 3.5 in \cite{gine2000exponential}]\label{lemma:gine2000}
	Let $f_{ij}(z_i, \widetilde{z}_j)$ be the kernel of a degenerate and decoupled second order $U$-statistic. Define
	\begin{align*}
		& A = \max_{1 \leq i, j \leq n} \sup_{z, \widetilde{z}} |f_{ij}(z, \widetilde{z})|, \quad B^2 = \max \left[ \sup_{\widetilde{z}} \sum_{i = 1}^n \E\left\{ f^2_{ij}(Z_i, \widetilde{z}) \right\}, \  \sup_{z} \sum_{j = 1}^n \E\left\{ f^2_{ij}(z, \widetilde{Z}_j) \right\} \right] \\
		& C^2 = \sum_{1 \leq i, j \leq n} \E \left\{f^2_{ij}(Z_i, \widetilde{Z}_j) \right\}
	\end{align*}
	where $\{Z_i,  1 \leq i \leq n\}$ are independent random variables and $\{\widetilde{Z}_j,  1 \leq j \leq n\}$ are independent copies of $Z_i$. Then, for a universal constant $K$, the following holds
	\begin{align*}
		\Pb\left(\left| \sum_{i, j} f_{ij}(z_i, \widetilde{z}_j) \right| > t \right) \leq K \exp\left[-\frac{1}{K}\min\left\{\frac{t}{C}, \left(\frac{t}{B}\right)^{2/3}, \left(\frac{t}{A} \right)^{1/2}\right\}\right]
	\end{align*}
\end{lemma}
Lemma \ref{lemma:gine2000} is for decoupled $U$-statistics, however, because of the result in \cite{de1995decoupling}, as noted in \cite{gine2000exponential} (pages 15 and 20), the same conclusion holds for regular, undecoupled $U$-statistics simply with $K$ replaced by a different constant. Thus, we will apply Lemma \ref{lemma:gine2000} without performing the additional decoupling step or introducing a different constant. In particular, we apply Lemma \ref{lemma:gine2000} with $f_{ij}(z_i, z_j) = \{n(n-1)\}^{-1}g_D(z_i, z_j)$ if  $i \neq j$ and zero otherwise; that is:
\begin{align*}
	& \Pb\left(\left| \frac{1}{n(n-1)}\sum_{i \neq j} g_D(z_i, z_j) \right| > t \mid D^n \right) \leq K \exp\left[-\frac{1}{K}\min\left\{\frac{t}{C}, \left(\frac{t}{B}\right)^{2/3}, \left(\frac{t}{A} \right)^{1/2}\right\}\right],\\
	&  \text{where } A = \frac{\sup_{z_1, z_2} |g_D(z_1, z_2)|}{n(n-1)}, \\
	& B^2 = \max \left[ \sup_{z_2} \frac{\E\left\{ g_D^2(Z_1, z_2) \mid D^n \right\}}{n^2(n-1)}, \  \sup_{z_1} \frac{\E\left\{ g_D^2(z_1, Z_2) \mid D^n \right\}}{n^2(n-1)}\right], \text{ and} \\
	& C^2 =  \frac{\E \left\{g_D^2(Z_1, Z_2) \mid D^n \right\}}{n(n-1)}
\end{align*}
\subsection{Bound on  $\Pb\left(\left|\mathbb{U}_n g_D(Z_1, Z_2) \right| > \frac{t}{12\sqrt{c_3J}} \mid D^n \right)$} 
First, notice that, under the assumption that $\|\widehat{Q}\| \lesssim 1$ and $\|Q^{-1}\| \lesssim 1$,
\begin{align*}
	[\widehat{Q}Q^{-1}R]_j & \lesssim \max_j|R_j| \\
	& = \max_j \left| \int \rho_{hj} (x) K_{h}(x) \var(A \mid X = x) \tau(x) f(x) dx \right| \\
	& = \max_j \left| \int \rho_j(1/2 + v) K(v) \var(A \mid X = x_0 + vh) \tau(x_0 + vh) f(x_0 + vh) dv\right| \\
	& \leq \max_j \sup_v|\rho_j(1/2 + v) \var(A \mid X = x_0 + vh) \tau(x_0 + vh) f(x_0 + vh)| \left|\int K(v) dv \right| \\
	& \lesssim 1
\end{align*}
and
\begin{align*}
	\left| \int g(z_1, z_2)d\Pb(z_1) \right| &  \leq \left|\int f_{2j}(z_1, z_2) d\Pb(z_1) \right| + [\widehat{Q}Q^{-1}R_j]_j \\&  \lesssim h^{-d} \one(2\|x_2 - x_0\| \leq h )\left|  \widehat\Pi \left(\frac{dF}{d\widehat{F}}(\pi - \widehat\pi)\right) (x_2) \right| + 1 \lesssim h^{-d}
\end{align*}
\begin{align*}
	\left| \int g(z_1, z_2)d\Pb(z_2) \right| & \leq \left|\int f_{2j}(z_1, z_2) d\Pb(z_2) \right| + [\widehat{Q}^{-1}QR_j]_j \\ &   \lesssim h^{-d}\one(2\|x_1 - x_0\| \leq h )\left| 1 - \widehat\Pi \left(\frac{dF}{d\widehat{F}}(\mu - \widehat\mu_0)\right) (x_1) \right| + 1 \lesssim h^{-d}
\end{align*}
and recall that by Proposition 9 in \cite{kennedy2022minimax}, 
\begin{align*}
	\left| \int g(z_1, z_2) d\Pb(z_1, z_2) \right| \leq c_4 \left(\frac{k}{h^{d}}\right)^{-2s/d}
\end{align*}
for some constant $c_4$. 

\textbf{Term A}.

We have
\begin{align*}
	\sup_{z_1, z_2} |g_D(z_1, z_2)| & \leq 4 \sup_{z_1, z_2} |g(z_1, z_2)| \\
	& \lesssim \sup_{z_1, z_2} \left|\rho_{hj}(x_1)K_{h}(x_1) \widehat\varphi_{y1}(z_1) + \rho_{hj}(x_1) K_h(x_1) \widehat\varphi_{y2}(z_1, z_2) K_{h}(x_2)\right|\\
	& \lesssim k h^{-2d}
\end{align*}
Thus, there exists a constant $c_A$ such that $A \lesssim \frac{k}{n(n-1)h^{2d}} \leq c_A\sqrt{ \frac{k}{n(n-1)h^{2d}}}$ since $kh^{-2d}n^{-2} \to 0$. 

\textbf{Term B}. 

We have
\begin{align*}
	& f_{2j}^2(z_1, z_2) \\
	& \lesssim h^{-2d}\one(2\|x_1 - x_0\| \leq h) \\
	& \hphantom{\lesssim} + h^{-2d}\one(2\|x_1 - x_0\| \leq h) h^{-2d}\one(2\|x_2 - x_0\| \leq h) b^T_h(X_1) \widehat\Omega^{-1} b_h(X_2) b^T_h(X_2) \widehat\Omega^{-1} b_h(X_1)
\end{align*}
and, for instance, for $dF^*(v) = dF(x_0 + h(v - 0.5))$:
\begin{align*}
	& \int  f_{2j}^2(z_1, z_2) d\Pb(z_1)  \lesssim h^{-d} \int_{v: \|v - 0.5\|\leq 0.5} dF^*(v) \\
	& \hphantom{\lesssim} + h^{-3d}\one(2\|x_2 - x_0\| \leq h) b_h(x_2)^T\widehat\Omega^{-1} \int_{v: \|v - 0.5\|\leq 0.5} b(v)b(v)^T dF^*(v)\widehat\Omega^{-1} b_h(x_2) \\
	& \lesssim \frac{k}{h^{3d}}
\end{align*}
and similarly for $\int  f_{2j}^2(z_1, z_2) d\Pb(z_2)$. Therefore,
\begin{align*}
	\sup_{z_2} \int g^2(z_1, z_2) d\Pb(z_1) & \lesssim \frac{k}{h^{3d}}  + [\widehat{Q}^{-1}QR_j]_j^2\lesssim \frac{k}{h^{3d}}
\end{align*}
and similarly for $\sup_{z_1} \int g^2(z_1, z_2) d\Pb(z_2)$. Furthermore, 
\begin{align*}
	& 	\sup_{z_2}\left| \int g(z_1, z_2)d\Pb(z_1) \right|  \lesssim h^{-d}, \quad \sup_{z_1} \left| \int g(z_1, z_2)d\Pb(z_2) \right|  \lesssim h^{-d}, \\
	&\text{and} \quad \left|\int g(z_1, z_2)d\Pb(z_1)d\Pb(z_2)\right| \lesssim \left(\frac{k}{h^{d}}\right)^{-2s/d}.
\end{align*}
Therefore,
\begin{align*}
	\sup_{z_2} \int g^2_D(z_1, z_2) d\Pb(z_1) \lesssim \frac{k}{h^{3d}} \quad \text{and} \quad 	\sup_{z_1} \int g^2_D(z_1, z_2) d\Pb(z_2) \lesssim \frac{k}{h^{3d}}.
\end{align*}
Thus, for some constant $c_B$, we have
\begin{align*}
	B \lesssim \frac{1}{\sqrt{nh^d}} \cdot \sqrt{\frac{k}{n(n-1)h^{2d}}} \leq c_B \sqrt{\frac{k}{n(n-1)h^{2d}}}
\end{align*} 
\textbf{Term C}. 

\begin{align*}
	& \int  f_{2j}^2(z_1, z_2) d\Pb(z_1) d\Pb(z_2)  \lesssim h^{-d} \int_{v: \|v - 0.5\|\leq 0.5} dF^*(v) \\
	& \hphantom{\lesssim} + h^{-2d} \int_{v:\|v - 0.5\|\leq0.5} b_h(v)^T\widehat\Omega^{-1} \int_{v: \|v - 0.5\|\leq 0.5} b(v)b(v)^T dF^*(v)\widehat\Omega^{-1} b_h(v) dF^*(v) \\
	& \lesssim \frac{k}{h^{2d}}
\end{align*}
so that
\begin{align*}
	\int g^2(z_1, z_2) d\Pb(z_1) d\Pb(z_2) & \lesssim  \frac{k}{h^{2d}}
\end{align*}
Furthermore, 
\begin{align*}
	\left| \int g(z_1, z_2)d\Pb(z_1)\right| \lesssim h^{-d} \quad \text{and} \quad \left| \int g(z_1, z_2)d\Pb(z_2) \right| \lesssim h^{-d}
\end{align*}
Therefore, 
\begin{align*}
	\int g_D^2(z_1, z_2)d\Pb(z_1)d\Pb(z_2)& \lesssim \int g^2(z_1, z_2) d\Pb(z_1) d\Pb(z_2) + \int \left\{\int g(z_1, z_2)d\Pb(z_1)\right\}^2d\Pb(z_2) \\
	& \hphantom{\lesssim} + \int \left\{\int g(z_1, z_2)d\Pb(z_2)\right\}^2d\Pb(z_1) \\
	& \lesssim \frac{k}{h^{2d}}
\end{align*}
This means that $C \leq c_C \sqrt{\frac{k}{n(n-1)h^{2d}}}$ for some constant $c_C$. 

To recap, we have derived that for some constant $K$, which now includes $c_A, c_B$ and $c_C$:
\begin{align*}
	& \Pb\left(\left|\Un g_D(Z_1, Z_2) \right| > \frac{t}{2\sqrt{c_3J}} \mid D^n \right) \\
	& \leq K \exp\left[-\frac{1}{K} \min \left\{ \frac{t}{\sqrt{\frac{k}{n(n-1)h^{2d}}}}, \left( \frac{t}{\sqrt{\frac{k}{n(n-1)h^{2d}}}}\right)^{2/3}, \left( \frac{t}{\sqrt{\frac{k}{n(n-1)h^{2d}}}}\right)^{1/2}\right\}\right],
\end{align*}
for $t \geq 4\sqrt{c_3J}c_4 \left(\frac{k}{h^d}\right)^{-2s/d}$.
\subsection{Bound on  $\Pb\left(\left|\Pn g_1(Z_1) \right| > \frac{t}{12\sqrt{c_3J}} \mid D^n \right)$} 
By Bernstein's inequality, we have
\begin{align*}
	\Pb\left(\left|\Pn g_1(Z_1) \right| > \frac{t}{12\sqrt{c_3J}} \mid D^n \right) \leq 2 \exp\left( -\frac{n(t/\{12\sqrt{c_3J}\})^2 / 2}{\E\{g^2_1(Z_1) \mid D^n\} + \sup_{z_1}(3)^{-1}|g_1(z_1)|t} \right)
\end{align*}
We have
\begin{align*}
	\sup_{z_1} \left| \int g(z_1, z_2)d\Pb(z_2) \right|  \leq c_5h^{-d}, \quad \text{and} \quad \left|\int g(z_1, z_2)d\Pb(z_1)d\Pb(z_2)\right| \leq c_4 \left(\frac{k}{h^{d}}\right)^{-2s/d}.
\end{align*}
Moreover,
\begin{align*}
	\int \left\{\int g(z_1, z_2) d\Pb(z_2) \right\}^2 d\Pb(z_1) \leq c_6h^{-d}
\end{align*}
Therefore, we have
\begin{align*}
	& \E\{g^2_1(Z_1) \mid D^n\} \leq 2c_6h^{-d} + 2 c_4 \left(\frac{k}{h^{d}}\right)^{-4s/d} \leq c_7 h^{-d} \quad \text{and} \\ & \sup_{z_1} |g_1(z_1)| \leq c_5h^{-d} + c_4 \left(\frac{k}{h^{d}}\right)^{-2s/d} \leq c_8h^{-d}
\end{align*}
since $k/h^d \to \infty$. Therefore, we conclude that for all $t \leq 3c_7 / c_8$ and $c_9 = 1/(4\cdot 12^2 c_3Jc_7)$:
\begin{align*}
	\Pb\left(\left|\Pn g_1(Z_1) \right| > \frac{t}{12\sqrt{c_3J}} \mid D^n \right) \leq 2\exp\left( - c_9nh^d t^2\right)
\end{align*}
\subsection{Bound on  $\Pb\left(\left|\Pn g_2(Z_2) \right| > \frac{t}{12\sqrt{c_3J}} \mid D^n \right)$} 
A similar to the one above yields that there exist constants $c_{10}$ and $\Delta_{10}$ such that
\begin{align*}
	\Pb\left(\left|\Pn g_2(Z_2) \right| > \frac{t}{12\sqrt{c_3J}} \mid D^n \right) \leq 2\exp\left( - c_{10}nh^d t^2\right) \text{ for all } t \leq \Delta_{10} 
\end{align*}
\subsection{Final step}
To conclude, notice that
\begin{align*}
	\sqrt{\frac{k}{n(n-1) h^{2d}}} = \sqrt{\frac{k}{n^2 h^{2d}} + {\frac{k}{n^2(n-1)h^{2d}}}} \lesssim  \sqrt{\frac{k}{n^2 h^{2d}}} \equiv a_n
\end{align*}
We have obtained that, for $x_0 \in D(\eta)$, there exists constants $K$, $\Delta$,  $c_{11}$ and $c_{12}$ such that, for all $c_{11}\max\left\{h^\gamma, \left(\frac{k}{h^d}\right)^{-2s/d} \right\} \leq t \leq \Delta$, it holds that:
\begin{align*}
	\Pb\left(\left|\widehat\tau(x_0) - \tau(x_0) > t\right| \mid D^n \right) & \leq K \exp \left[-\frac{1}{K} \min \left\{\frac{t}{a_n}, \left(\frac{t}{a_n}\right)^{2/3}, \left(\frac{t}{a_n}\right)^{1/2} \right\} \right] \\
	& \hphantom{\leq} + 4\exp\left(-c_{12} nh^d t^2\right)
\end{align*}
The optimal choice of $k$ and $h$ depends on the values of $\gamma$, $s$ and $d$. In particular, we distinguish two cases.

\textbf{Case I: $s \geq \frac{d/4}{1 + d/2\gamma}$}. Set
\begin{align*}
	h = n^{-1/(2\gamma + d)} \text{ and } k = nh^d = n^{2\gamma/(2\gamma + d)} \implies a_n = \frac{1}{\sqrt{nh^d}} =  n^{-\gamma/(2\gamma + d)}.
\end{align*} 
In this case, we have for all $t$ such that $c_a a_n \leq t \leq \Delta$ for some constant $c_a$:
\begin{align*}
	\Pb\left(\left|\widehat\tau(x_0) - \tau(x_0) > t\right| \mid D^n \right) \leq K \exp \left[-\frac{1}{K} \min \left\{\left(\frac{t}{a_n}\right)^2, \frac{t}{a_n}, \left(\frac{t}{a_n}\right)^{2/3}, \left(\frac{t}{a_n}\right)^{1/2} \right\} \right]
\end{align*}
for some constant $K$.

\textbf{Case II: $s < \frac{d/4}{1 + d/2\gamma}$}. Define $T = 1 + d/(2\gamma) + d/(4s)$, where we recall $s = (\alpha + \beta) / 2$ ($\alpha$ = smoothness of $\pi, \widehat\pi$ and $\beta$ = smoothness of $\mu, \widehat\mu$). Set
\begin{align*}
	& h = n^{-1/(T\gamma)} \text{ and } k = n^{ \{d/(2s) - d/\gamma\} / T} \\ & \implies a_n = n^{-1/T}  \text{ and } a_n \geq \frac{1}{\sqrt{nh^d}}= n^{d/(2T\gamma) - 1/2} \text{ because } s < \frac{d/4}{1 + d/2\gamma}. 
\end{align*}
Thus, in this case too, we have for all $t$ such that $c_a a_n \leq t \leq \Delta$ for some constant $c_a$:
\begin{align*}
	\Pb\left(\left|\widehat\tau(x_0) - \tau(x_0) > t\right| \mid D^n \right) \leq K \exp \left[-\frac{1}{K} \min \left\{\left(\frac{t}{a_n}\right)^2, \frac{t}{a_n}, \left(\frac{t}{a_n}\right)^{2/3}, \left(\frac{t}{a_n}\right)^{1/2} \right\} \right]
\end{align*}
for some constant $K$. 

If $x_0 \not \in D(\eta)$, the same arguments use to prove the $x_0 \in D(\eta)$ case hold simply with $\gamma$ replaced by $\gamma'$. The final rate would be $a'_n = n^{-1/T'}$, where $T' = 1 + d/(2\gamma') + d/(4s)$. 
\section{Proof of Theorem \ref{thm:lower_bound}}
We prove the case when $\alpha \geq \beta$, since the case $\alpha < \beta$ can be proved in a symmetric way by swapping the perturbations of $\mu_0$ and $\pi$ in a way analogous to that presented in \cite{kennedy2022minimax}. We proceed as follows
\begin{enumerate}
	\item Let $z = (y, a, x) \in \{0, 1\}^2 \times [0, 1]^d$ and $f(x)$ the density of $x$ with respect to the Lebesgue measure. We consider a data generating process such that each observation follows a distribution with density $p_{\omega, \lambda}(z)$, where $\omega \in \Omega = \{0, 1\}^m$ and $\lambda = \{0, 1\}^{2mk}$, for some $k$ and with prior $\overline\nu$ on $\lambda$. The sample of $n$ independent observations has thus density
	\begin{align*}
		p^n_\omega \equiv p^n_\omega(z_1, \ldots, z_n) = \int \prod_{i = 1}^n p_{\omega, \lambda}(z_i) d\overline\nu(\lambda)
	\end{align*} 
	Depending on $\omega$, the density will have $\tau_\omega(x) = \mu_1(x) - \mu_0(x)$ fluctuated. For each density, the $\lambda$ vector will govern the fluctuations of $\pi(x)$ and $\mu_0(x)$ and will not generally interact with $\omega$. Notice that we parametrize the density by $(\pi, \mu_0, \tau)$, so that $\mu_1(x) = \tau(x) + \mu_0(x)$.
	Crucially, we will establish that $p_{\omega, \lambda}(z)$ belongs to $\mathcal{P}$, the set of all densities compatible with assumptions \ref{assumption_margin} and \ref{def:smoothness}, so that we have
	\begin{align*}
		\inf_{\widehat\Gamma} \sup_{p \in \mathcal{P}} \E\{d_H(\widehat\Gamma, \Gamma_p)\} \geq \inf_{\widehat\Gamma} \max_{\omega} \E_{p^n_\omega}\{d_H(\widehat\Gamma, \Gamma_{\omega})\}
	\end{align*}
	where $\Gamma_\omega$ is the true upper level set when the data is sampled from $p_\omega^n$. 
	\item Under the margin assumption \ref{assumption_margin}, we rely Proposition 2.1 in \cite{rigollet2009optimal} to obtain
	\begin{align*}
		d_H(\widehat\Gamma, \Gamma_{\omega}) \gtrsim \left[\int_{(\widehat\Gamma \Delta \Gamma_{\omega}) \cap \{\tau_{\omega}(x) \neq \theta\} } f(x) dx \right]^{(1 + \xi)/ \xi}
	\end{align*}
	\item We construct a vector $\widehat\omega$ such that
	\begin{align*}
		\int_{(\widehat\Gamma \Delta \Gamma_{\omega}) \cap \{\tau_{\omega}(x) \neq \theta \}} f(x) dx \geq \frac{1}{2} \int_{(\Gamma_{\widehat\omega} \Delta \Gamma_{\omega}) \cap \{\tau_{\omega}(x) \neq \theta \} } f(x) dx
	\end{align*}
	and show that
	\begin{align*}
		\int_{(\Gamma_{\widehat\omega} \Delta \Gamma_{\omega}) \cap \{\tau_{\omega}(x) \neq \theta \}} f(x) dx = 2\text{Leb}_d\{\mathcal{S}_{hk}(x_1)\}\rho(\widehat\omega, \omega)
	\end{align*}
	where $\rho(\widehat\omega, \omega) = \sum_{i = 1}^n \one(\widehat\omega_i \neq \omega_i)$ is the Hamming distance and $\mathcal{S}_{hk}(x_1)$ is a particular set defined below. 
	
	At this point we have the following chain of inequalities:
	\begin{align*}
		\inf_{\widehat\Gamma} \sup_{p \in \mathcal{P}} \E\{d_H(\widehat\Gamma, \Gamma_p)\} & \geq \inf_{\widehat\Gamma} \max_{\omega} \E_{p_\omega^n}\{d_H(\widehat\Gamma, \Gamma_{\omega})\} \\
		& \gtrsim \inf_{\widehat\Gamma} \max_{\omega} \E_{p_\omega^n}\left(\left[\int_{(\widehat\Gamma \Delta \Gamma_{\omega}) \cap \{\tau_{\omega}(x) \neq \theta\}} f(x) dx \right]^{(1 + \xi)/ \xi}\right) \\
		& \geq \inf_{\widehat\Gamma} \max_{\omega} \left(\E_{p_n^\omega}\left[\int_{(\widehat\Gamma \Delta \Gamma_{\omega}) \cap \{\tau_{\omega}(x) \neq \theta\}} f(x) dx \right]\right)^{(1 + \xi)/ \xi} \\
		& \geq \frac{1}{2}\inf_{\widehat{\omega}} \max_{\omega} \left(\E_{p_n^\omega}\left[\int_{(\Gamma_{\widehat\omega} \Delta \Gamma_{\omega}) \cap \{\tau_{\omega}(x) \neq \theta \}} f(x) dx \right]\right)^{(1 + \xi)/ \xi} \\
		& = \left[\text{Leb}_d\{\mathcal{S}_{hk}(x_1)\}\right]^{(1+\xi)/\xi}\inf_{\widehat{\omega}} \max_{\omega} \left[\E_{p_n^\omega}\left\{\rho(\widehat\omega, \omega) \right\}\right]^{(1 + \xi)/ \xi}
	\end{align*}
	\item By Theorem 2.12 in \cite{tsybakov2004introduction}, if the Hellinger distance satisfies $H^2(p^n_{\omega'}, p^n_{\omega}) \leq 1$ for any $\omega', \omega$ such that $\rho(\omega', \omega) =1$, then
	\begin{align*}
		\inf_{\widehat\omega}\max_{\omega} \E_{p^n_\omega}\rho(\widehat\omega, \omega) \geq m\left(\frac{1}{2} - \frac{\sqrt{3}}{4} \right)
	\end{align*}
	We show that $\text{Leb}_d\{\mathcal{S}_{hk}(x_1)\} = (h/2)^d$ so that, putting everything together, we have
	\begin{align*}
		\inf_{\widehat\Gamma} \sup_{p \in \mathcal{P}} \E\{d_H(\widehat\Gamma, \Gamma_p)\} \gtrsim (h^d m)^{(1 + \xi) / \xi} 
	\end{align*}
	Choosing $h = O\left(n^{-1/(T \gamma)}\right)$ and $m = O\left(h^{-d + \gamma \xi}\right)$ yields the desired rate, where $T = 1 + d/(4s) + d/(2\gamma)$. 
	\item We verify that choosing $h = O\left(n^{-1 / (T\gamma)}\right)$, $m = O\left( h^{-d + \gamma \xi}\right)$ and $k = O\left(n^{d(\gamma - 2s) / (2s\gamma T) }\right)$ yields $H^2(p^n_{\omega'}, p^n_{\omega}) \leq 1$ for any $\omega', \omega$ such that $\rho(\omega', \omega) =1$. 
\end{enumerate}
\subsubsection{Step 1: Construction of fluctuated densities}
Let $x_1, \ldots, x_{2m}$ denote a grid of $[0, 1]^d$, for some $m$ to be chosen later, and $\mathcal{C}_{h}(x_i)$ a cube with side $h$ centered at $x_i$.  Let $\mathcal{C}_{h / k^{1/d}}(m_{ji})$, $j \in \{1, \ldots, k\}$, be a partition of the cube $\mathcal{C}_h(x_i)$ into $k$, equally-sized cubes with midpoints $m_{1i}, \ldots, m_{ki}$. Then, for $\lambda \in \{-1, 1\}^{2mk}$ and $\omega \in\{0, 1\}^{m}$, define the functions
\begin{align*}
	& \tau_{\omega}(x) = \theta + h^\gamma \sum_{i = 1}^m \left[\omega_i B\left( \frac{x - x_i}{h} \right) + (1 - \omega_i) B\left( \frac{x - x_{i + m}}{h} \right)\right]\\
	& \mu_{0\lambda} (x) = \frac{1}{2} + \left(\frac{h}{k^{1/d}}\right)^{\beta} \sum_{i = 1}^{2m} \sum_{j = 1}^k \lambda_{ij} B\left( \frac{x - m_{ji}}{h / 2k^{1/d}} \right) - \frac{\tau_{\omega}(x)}{2} \\
	& \pi_{\lambda\omega}(x) = \frac{1}{2} + \left(\frac{h}{k^{1/d}}\right)^{\alpha} \sum_{i = 1}^m \sum_{j = 1}^k \left\{(1 - \omega_i)\lambda_{ij}B\left( \frac{x - m_{ji}}{h / 2k^{1/d}} \right) + \omega_i \lambda_{i+mj} B\left( \frac{x - m_{ji + m}}{h / 2k^{1/d}} \right) \right\}\\
	& f(x) = c_{hm}\left[ 1 - \sum_{i = 1}^{2m}\one\{x \in \mathcal{C}_{2h}(x_i)\} + \sum_{i = 1}^{2m} \one\{x \in \mathcal{S}_{hk}(x_i)\} \right]\\
	& \mathcal{S}_{hk}(x_i) = \cup_{j = 1}^k \mathcal{C}_{h/2k^{1/d}} (m_{ji})
\end{align*}
where $c_{hm} = \{1 - 2(2^d - 2^{-d})h^dm\}^{-1}$ and $B(u)$ is an infinitely differentiable function such that $B(u) = 1$ for $u \in [-1/2, 1/2]^d$ and $B(u) = 0$ for $u \not\in [-1, 1]^d$. For example, $B\left( \frac{x - m_{ji}}{h / 2k^{1/d}} \right) = 0$ for any $x \not\in \mathcal{C}_{h/k^{1/d}}(m_{ji})$ and $B\left( \frac{x - m_{ji}}{h / 2k^{1/d}} \right) = 1$ for any $x \in \mathcal{C}_{h/2k^{1/d}} (m_{ji})$. 

Also notice that $f(x) = c_{hm}$ for any $x \not\in \cup_{i = 1}^{2m} \mathcal{C}_{2h}(x_i) \equiv \mathcal{X}_0$. In this region $\mathcal{X}_0$, $\tau_\omega(x) = \theta$, $\mu_{0\lambda}(x) = (1-\theta) / 2$ and$\pi_{\omega\lambda}(x) = 1/2$. Therefore, the density of each observation can be written as 
\begin{align*}
	& p_{\omega, \lambda}(z) = c_{hm}  \sum_{i = 1}^m \one\{x \in \mathcal{S}_{hk}(x_i)\} \{\omega_i p_{i \lambda}(z) + (1 - \omega_i)q_{i\lambda}(z)\} \\
	& \hphantom{p_{\omega}(z) = c_{hm}\sum_{i = 1}^m} + \one\{x \in \mathcal{S}_{hk}(x_{i + m})\} \{ (1 - \omega_i) p_{i \lambda}(z) + \omega_i q_{i\lambda}(z)\} \\ 
	& \hphantom{p_{\omega}(z) =} + \frac{1}{4} c_{hm} \one\{x \in \mathcal{X}_0\} \left\{ 1 + (2y - 1)(2a - 1)\theta \right\}
\end{align*}
where, for $s = (\alpha + \beta) / 2$,
\begin{align*}
	& p_{i\lambda}(z) = \frac{1}{4} + (y - 1/2) \left(\frac{h}{k^{1 / d}}\right)^\beta \sum_{j = 1}^k \lambda_{ij}B\left( \frac{x - m_{ji}}{h/2k^{1/d}} \right) +  (2a - 1)(2y - 1)\left\{\frac{\theta}{4} + \frac{h^\gamma}{4} B\left( \frac{x - x_i}{h} \right) \right\} \\
	& q_{i\lambda}(z) = \frac{1}{4} + \left[(a - 1/2)\left(\frac{h}{k^{1 / d}}\right)^\alpha + (y - 1/2)\left\{\left(\frac{h}{k^{1 / d}}\right)^\beta + \theta \left(\frac{h}{k^{1 / d}}\right)^{\alpha} \right\} \right]\sum_{j = 1}^k \lambda_{ij}B\left( \frac{x - m_{ji}}{h/2k^{1/d}} \right) \\
	& \hphantom{q_{i\lambda}(z)} \quad  + (2a - 1)(2y - 1) \left\{ \frac{\theta}{4} + \left(\frac{h}{k^{1 / d}}\right)^{2s}  \sum_{j = 1}^kB\left( \frac{x - m_{ji}}{h/2k^{1/d}} \right)^2  \right\}
\end{align*}
It is possible to verify that $\tau_\omega(x)$ is $\gamma$-smooth, $\mu_{0\lambda}(x)$ is $\beta$-smooth and $\pi_{\lambda\omega}(x)$ is $\alpha$-smooth. To verify the margin condition, notice that, for any $t > 0$:
\begin{align*}
	\int_{x \in \mathcal{X}: 0 < |\tau_{\omega}(x) - \theta| < t} f(x) dx & = c_{hm} \sum_{i = 1}^m  \text{Leb}_d\{x \in \mathcal{S}_{hk}(x_i) \cap 0 < |\tau_\omega(x) - \theta| < t\} \\
	& = c_{hm}m \text{Leb}_d\{x \in \mathcal{S}_{hk}(x_1) \cap 0 < |\tau_\omega(x) - \theta| < t\} \\
	& =  \frac{c_{hm}}{2^d}h^dm \one(t > h^{\gamma}) \\
	& = \frac{c_{hm}}{2^d} h^{\gamma \xi} \one(t^\xi > h^{\xi\gamma}) \\
	& \lesssim  t^\xi 
\end{align*}
because the choice of $h$ and $m$ ensures that $c_{hm}$ is upper bounded by a constant. 
\subsubsection{Step 2: Proposition 2.1 in \cite{rigollet2009optimal}}
By Proposition 2.1 in \cite{rigollet2009optimal}, under the margin assumption \ref{assumption_margin}, we have
\begin{align*}
	d_H(G_1, G_2)^{\xi / (1 + \xi)} \gtrsim \int_{G_1 \Delta G_2 \cap \tau(x) \neq \theta} f(x) dx
\end{align*}
for any $G_1$ and $G_2$. 
\subsubsection{Step 3: Reduction from $\widehat\Gamma$ to $\Gamma_{p_{\widehat\omega}}$}
Define
\begin{align*}
	\widehat\omega_i = \begin{cases}
		0 & \text{ if } \text{Leb}_d\{\widehat\Gamma \cap \mathcal{S}_{hk} (x_i)\} <  \text{Leb}_d\{\widehat\Gamma \cap \mathcal{S}_{hk}(x_{i + m})\} \\
		1 & \text{ otherwise}
	\end{cases}
\end{align*}
and notice that
\begin{align*}
	\int_{(\widehat\Gamma \Delta \Gamma_{\omega}) \cap \tau_\omega(x) \neq \theta} f(x) dx & = c_{hm}  \sum_{i = 1}^{m} \text{Leb}_d [\widehat\Gamma \Delta \Gamma_{\omega} \cap \{\mathcal{S}_{hk}(x_i) \cup \mathcal{S}_{hk}(x_{i + m})\}] \\
	& \geq \sum_{i = 1}^{m} \text{Leb}_d [\widehat\Gamma \Delta \Gamma_{\omega} \cap \{\mathcal{S}_{hk}(x_i) \cup \mathcal{S}_{hk}(x_{i + m})\}] \\
	& \geq \frac{1}{2} \sum_{i = 1}^{m} \text{Leb}_d [\Gamma_{\widehat\omega}\Delta \Gamma_{\omega} \cap \{\mathcal{S}_{hk}(x_i) \cup \mathcal{S}_{hk}(x_{i + m})\}] 
\end{align*}
The second inequality follows because $c_{hm} \geq 1$. To see why the third inequality holds, first consider the case where $\omega_i = 1$. If $\widehat\omega_i = 1$, then 
\begin{align*}
	0 = \text{Leb}_d [\Gamma_{\widehat\omega}\Delta \Gamma_{\omega} \cap \{\mathcal{S}_{hk}(x_i) \cup \mathcal{S}_{hk}(x_{i + m})\}] \leq \text{Leb}_d [\widehat\Gamma \Delta \Gamma_{\omega} \cap \{\mathcal{S}_{hk}(x_i) \cup \mathcal{S}_{hk}(x_{i + m})\}]  
\end{align*}
If $\widehat\omega_i = 0$, then it means that $\text{Leb}_d\{\widehat\Gamma \cap \mathcal{S}_{hk} (x_i)\} <  \text{Leb}_d\{\widehat\Gamma \cap \mathcal{S}_{hk}(x_{i + m})\}$ so that
\begin{align*}
	& \text{Leb}_d [\widehat\Gamma \Delta \Gamma_{\omega} \cap \{\mathcal{S}_{hk}(x_i) \cup \mathcal{S}_{hk}(x_{i + m})\}]  = \text{Leb}_d \{\widehat\Gamma^c \cap \mathcal{S}_{hk}(x_i)\} + \text{Leb}_d \{\widehat\Gamma \cap \mathcal{S}_{hk}(x_{i + m})\} \\
	& = \text{Leb}_d \{\mathcal{S}_{hk}(x_i)\} - \text{Leb}_d \{\widehat\Gamma \cap \mathcal{S}_{hk}(x_i)\} + \text{Leb}_d \{\widehat\Gamma \cap \mathcal{S}_{hk}(x_{i + m})\} \\
	& > \text{Leb}_d \{\mathcal{S}_{hk}(x_i)\} \\
	& = \frac{1}{2} \text{Leb}_d [\Gamma_{\widehat\omega}\Delta \Gamma_{\omega} \cap \{\mathcal{S}_{hk}(x_i) \cup \mathcal{S}_{hk}(x_{i + m})\}] 
\end{align*}
Similarly, consider the case where $\omega_i = 0$. If $\widehat\omega_i = 0$, then as before
\begin{align*}
	0 = \text{Leb}_d [\Gamma_{\widehat\omega}\Delta \Gamma_{\omega} \cap \{\mathcal{S}_{hk}(x_i) \cup \mathcal{S}_{hk}(x_{i + m})\}] \leq \text{Leb}_d [\widehat\Gamma \Delta \Gamma_{\omega} \cap \{\mathcal{S}_{hk}(x_i) \cup \mathcal{S}_{hk}(x_{i + m})\}]  
\end{align*}
If $\widehat\omega_i = 1$, then it means that $\text{Leb}_d\{\widehat\Gamma \cap \mathcal{S}_{hk} (x_i)\} \geq \text{Leb}_d\{\widehat\Gamma \cap \mathcal{S}_{hk}(x_{i + m})\}$ so that
\begin{align*}
	&	\text{Leb}_d [\widehat\Gamma \Delta \Gamma_{\omega} \cap \{\mathcal{S}_{hk}(x_i) \cup \mathcal{S}_{hk}(x_{i + m})\}] = \text{Leb}_d \{\widehat\Gamma \cap \mathcal{S}_{hk}(x_i)\} + \text{Leb}_d\{\widehat\Gamma^c \cap \mathcal{S}_{hk}(x_{i + m})\} \\
	& = \text{Leb}_d \{\widehat\Gamma \cap \mathcal{S}_{hk}(x_i)\} + \text{Leb}_d \{\mathcal{S}_{hk}(x_{i + m})\} -  \text{Leb}_d \{\widehat\Gamma \cap \mathcal{S}_{hk}(x_{i + m})\} \\
	& \geq \text{Leb}_d \{\mathcal{S}_{hk}(x_{i + m})\} \\
	& = \frac{1}{2} \text{Leb}_d [\Gamma_{\widehat\omega}\Delta \Gamma_{\omega} \cap \{\mathcal{S}_{hk}(x_i) \cup \mathcal{S}_{hk}(x_{i + m})\}] 
\end{align*}
We have
\begin{align*}
	\sum_{i = 1}^{m} \text{Leb}_d [\Gamma_{\widehat\omega}\Delta \Gamma_{\omega} \cap \{\mathcal{S}_{hk}(x_i) \cup \mathcal{S}_{hk}(x_{i + m})\}] = 2\text{Leb}_d\{\mathcal{S}_{hk}(x_1)\}\sum_{i = 1}^{m} \one(\widehat\omega_i \neq \omega_i)
\end{align*}
Therefore, we have that, for any $\widehat\Gamma$, it holds that
\begin{align*}
	\max_{\omega \in \Omega} \E_\omega\left\{\int_{(\widehat\Gamma \Delta \Gamma_{\omega}) \cap \tau_\omega(x) \neq \theta} f(x) dx \right\} & \geq \text{Leb}_d\{\mathcal{S}_{hk}(x_1)\}\inf_{\widehat\omega}\max_{\omega \in \Omega}\E_{p^n_\omega}\rho(\widehat\omega, \omega)
\end{align*}
\subsubsection{Step 4: Final bound}
By Theorem 2.12 in \cite{tsybakov2004introduction}, if we can show that $H^2(p^n_\omega, p^n_{\omega'}) \leq 1$ for any $\omega, \omega' \in \Omega$ such that $\rho(\omega', \omega) = 1$, then
\begin{align*}
	\inf_{\widehat\omega}\max_{\omega \in \Omega}\E_\omega\rho(\widehat\omega, \omega) \geq m\left(\frac{1}{2} - \frac{\sqrt{3}}{4}\right)
\end{align*}
Thus, 
\begin{align*}
	\inf_{\widehat\Gamma} \sup_{p \in \mathcal{P}} \E\{d_H(\widehat\Gamma, \Gamma_p)\}  & \gtrsim  (mh^d)^{{(1 + \xi) / \xi}} 
\end{align*}
because $\text{Leb}_d\{\mathcal{S}_{hk}(x_1)\} = 2^{-d}h^d$. By choosing $h = O\left( n^{-1/(T \gamma)}\right)$ and $m = O\left( h^{-d + \gamma \xi}\right)$, where
\begin{align*}
	T = 1 + d/ (4s) + d/(2\gamma),
\end{align*}
we get a lower bound of order $n^{-(1 + \xi) / T}$ as desired.

\subsubsection{Step 5: Verification of upper bound on Hellinger distance}
\begin{lemma}[\cite{robins2009semiparametric}, \cite{kennedy2022minimax}] \label{lemma:robins}
	Let $P_\lambda$ and $Q_\lambda$ denote distributions indexed by a vector $\lambda = (\lambda_1, \ldots, \lambda_k)$, and let $\mathcal{Z} = \cup_{j = 1}^k \mathcal{Z}_j$ denote a partition of the sample space. Assume:
	\begin{enumerate}
		\item $P_\lambda(\mathcal{Z}_j) = Q_\lambda(\mathcal{Z}_j) = p_j$ for all $\lambda$, and
		\item the conditional distributions $\one_{\mathcal{Z}_j} dP_\lambda / p_j$ and $\one_{\mathcal{Z}_j} dQ_\lambda / p_j$ do not depend on $\lambda_l$ for $l \neq j$, and only differ on partitions $j \in S \subseteq \{1, \ldots, k\}$.
	\end{enumerate}
	For a prior distribution $\overline\nu$ over $\lambda$, let $\overline{p} = \int p_\lambda d \overline\nu(\lambda)$ and $\overline{q} = \int q_\lambda d \overline\nu(\lambda)$, and define
	\begin{align*}
		& \delta_1 = \max_{j \in S} \sup_\lambda \int_{\mathcal{Z}_j} \frac{(p_\lambda - \overline{p})^2}{p_\lambda p_j} d\nu \\
		& \delta_2 =  \max_{j \in S} \sup_\lambda \int_{\mathcal{Z}_j} \frac{(q_\lambda -p_\lambda)^2}{p_\lambda p_j} d\nu \\
		& \delta_3 =  \max_{j \in S} \sup_\lambda \int_{\mathcal{Z}_j} \frac{(\overline{q} - \overline{p})^2}{p_\lambda p_j} d\nu
	\end{align*}
	for a dominating measure $\nu$. If $\overline{p} / p_\lambda \leq b < \infty$ and $n p_j \max (1, \delta_1, \delta_2)\leq b$ for all $j$, then
	\begin{align*}
		H^2\left(\int P_\lambda^n d \overline\nu(\lambda), Q_\lambda^n d\overline\nu(\lambda)\right) \leq Cn\sum_{j \in S} p_j \left\{ n\left(\max_{j \in S} p_j\right)(\delta_1\delta_2 + \delta_2^2) + \delta_3\right\}
	\end{align*}
	for a constant $C$ only depending on $b$. 
\end{lemma}
It remains to verify that, given our choices of $h$, $m$, and $k$, it holds that $H^2(p^n_\omega, p^n_{\omega'}) \leq 1$ for any $\omega, \omega' \in \Omega$ such that $\rho(\omega', \omega) = 1$.

Following similar calculations as \cite{kennedy2022minimax}, we will rely on their Lemma 2 (from \cite{robins2009semiparametric} and restated above in Lemma \ref{lemma:robins}) to derive a bound on the Hellinger distance. 

Let us partition the space according to $\mathcal{Z}_{ji} = \mathcal{C}_{h / 2k^{1/d}} (m_{ji}) \times \{0, 1\}^2$, $j \in \{1, \ldots, k\}$ and $i \in \{1, \ldots, 2m\}$ and $\mathcal{Z}_0 = ([0, 1]^d \times \{0, 1\}^2) / \left( \cup_i \cup_j \mathcal{Z}_{ji} \right)$. On $\mathcal{Z}_0$, we have for any $\omega$:
\begin{align*}
	p_\omega(z) = \frac{1}{4} c_{hm} \one(x \in \mathcal{X}_0) \{1 + (2y - 1)(2a - 1) \theta\} \text{ for any } z \in \mathcal{Z}_0,
\end{align*} 
so that $\int_{\mathcal{Z}_0} p_\omega(z) dz = c_{hm}(1 - 2^{d + 1} m h^d)$.

Next, notice that, because $\rho(\omega', \omega) = 1$, the densities $p^n_{\omega'}$ and $p^n_\omega$ differ only on two cubes, which, without loss of generality, we take to be $\mathcal{C}_{2h}(x_{1})$ and $\mathcal{C}_{2h}(x_{m + 1})$. This corresponds to a difference in the first coordinate of $\omega$. To keep things clear, let $\omega_1 = (1, \omega_2, \ldots, \omega_m)$ and $\omega_0 = (0, \omega_2, \ldots, \omega_m)$. This way, we have
{\tiny
	\begin{align*}
		& p_{\omega_1, \lambda}(z)  = \\
		& \quad c_{hm} \one\{x \in \mathcal{S}_{hk}(x_1)\} \left[ \frac{1}{4} + (y - 1/2) \left(\frac{h}{k^{1/d}}\right)^{\beta}\sum_{j = 1}^k \lambda_{1j}B\left( \frac{x - m_{j1}}{h/2k^{1/d}} \right) +  (2a - 1)(2y - 1)\left\{\frac{\theta}{4} + \frac{h^\gamma}{4} B\left( \frac{x - x_1}{h} \right) \right\}\right] \\
		& \quad \hphantom{=} +  c_{hm}\sum_{i= 2}^m \one\{x \in \mathcal{S}_{hk}(x_i)\} \{\omega_i p_{i \lambda}(z) + (1 - \omega_i)q_{i\lambda}(z)\}  + \one\{x \in \mathcal{S}_{hk}(x_{i + m})\} \{ (1 - \omega_i) p_{i \lambda}(z) + \omega_i q_{i\lambda}(z)\} \\ 
		& \quad \hphantom{=} + c_{hm} \one\{x \in \mathcal{S}_{hk}(x_{1+m})\}  \left(\frac{1}{4} + \left[(a - 1/2)\left(\frac{h}{k^{1/d}}\right)^{\alpha} + (y - 1/2)\left\{\left(\frac{h}{k^{1/d}}\right)^{\beta} + \theta  \left(\frac{h}{k^{1/d}}\right)^{\alpha}\right\}\right] \sum_{j = 1}^k \lambda_{1+mj}B\left( \frac{x - m_{j1+m}}{h/2k^{1/d}} \right)\right. \\
		& \quad \hphantom{c_{hm} \one\{x \in \mathcal{S}_{hk}(x_{1+m})\}  } \quad\quad\quad \left. + (2a - 1)(2y - 1) \left\{ \frac{\theta}{4} + \left(\frac{h}{k^{1/d}}\right)^{2s}  \sum_{j = 1}^kB\left( \frac{x - m_{j1+m}}{h/2k^{1/d}} \right)^2 \right\}\right) \\
		&\quad \hphantom{=}  + \frac{1}{4} c_{hm} \one\{x \in \mathcal{X}_0\} \left\{ 1 + (2y - 1)(2a - 1) \theta \right\}
	\end{align*}
}%
and $p_{\omega_0, \lambda}(z)$ similarly defined. We will apply Lemma \ref{lemma:robins} with $P_\lambda(z) = p_{\omega_1, \lambda}(z)$ and $Q_\lambda(z) = p_{\omega_0, \lambda}(z)$. 

First, notice that, for any $(i, j)$ and vector $\lambda$, we have
\begin{align*}
	\int_{\mathcal{Z}_{ji}} p_{\omega_1, \lambda}(z) dz = \int_{\mathcal{Z}_{ji}} p_{\omega_0, \lambda}(z) dz = \frac{c_{hm}h^d}{2^dk} \equiv p_{ji}
\end{align*}
Further, $\lambda_{ij}$ only affects the densities in $\mathcal{Z}_{ji}$, so the second condition in the lemma is satisfied.

Furthermore, because  $p_{\omega_1, \lambda}(z)$ only differs from $p_{\omega_0, \lambda}(z)$ on $2k$ elements of the partition, it holds that
\begin{align*}
	\sum_{(ij) \in S} p_{ji} = 	\sum_{j = 1}^{k} p_{1j} + p_{1+mj} \lesssim h^d
\end{align*}
Therefore, provided we can verify the other assumptions of Lemma \ref{lemma:robins}, the Hellinger distance is upper bounded by
\begin{align*}
	H^2\left(\int p^n_{\omega_1, \lambda} d \overline\nu(\lambda), p^n_{\omega_0, \lambda} d\overline\nu(\lambda)\right) \lesssim n^2 h^d \left(\max_{(ji) \in S}p_{ji} \right)(\delta_1\delta_2 + \delta_2^2) + nh^d\delta_3
\end{align*}
We take $\overline\nu(\lambda)$ to be a uniform prior on $\lambda$ so that $\lambda_j = \{-1, 1\}$ independently and with equal probability. Then,
\begin{align*}
	&	\overline{p}_{i}(z)  \equiv \int p_{i\lambda}(z) d\overline\nu(\lambda) = \frac{1}{4} + (2a - 1)(2y - 1)\left\{ \frac{\theta}{4} + \frac{h^\gamma}{4} B\left( \frac{x - x_i}{h} \right) \right\} \\
	& \overline{q}_{i}(z) \equiv \int q_{i\lambda}(z) d\overline\nu(\lambda)= \frac{1}{4} + (2a - 1)(2y - 1) \left[\frac{\theta}{4} + \left(\frac{h}{k^{1/d}}\right)^{2s} \sum_{j = 1}^k B\left( \frac{x - m_{ji}}{h / 2k^{1/d}} \right)^2\right] \\
	& \overline{p}_{-1}(z) \equiv c_{hm} \sum_{i = 2}^m \one\{x \in \mathcal{S}_{hk}(x_i)\} \{\omega_i \overline{p}_{i}(z) + (1 - \omega_i)\overline{q}_{i}(z)\} \\
	& \quad\quad + \one\{x \in \mathcal{S}_{hk}(x_{i + m})\} \{ (1 - \omega_i) \overline{p}_{}(z) + \omega_i \overline{q}_{i}(z)\} + \frac{1}{4} c_{hm} \one\{x \in \mathcal{X}_0\} \left\{ 1 + (2y - 1)(2a - 1) \theta \right\} \\
	& \overline{p}_{\omega_1}(z) \equiv \int p_{\omega_1, \lambda}(z) d\overline\nu(\lambda) =  \overline{p}_{-1}(z) + c_{hm} \left[ \one\{x \in \mathcal{S}_{hk}(x_{1})\} \overline{p}_{1}(z) + \one\{x \in \mathcal{S}_{hk}(x_{1 + m})\} \overline{q}_{1+m}(z) \right]\\
	& \overline{p}_{\omega_0}(z) \equiv \int p_{\omega_0, \lambda}(z) d\overline\nu(\lambda) =  \overline{p}_{-1}(z) +  c_{hm} \left[\one\{x \in \mathcal{S}_{hk}(x_{1})\} \overline{q}_{1}(z) + \one\{x \in \mathcal{S}_{hk}(x_{1 + m})\} \overline{p}_{1+m}(z)\right]
\end{align*}
Next, we bound
\begin{align*}
	\delta_1 \equiv \max_{(ij) \in S}\sup_{\lambda} \int_{\mathcal{Z}_{ji}} \frac{\{p_{\omega_1, \lambda}(z) - \overline{p}_{\omega_1}(z)\}^2}{p_{\omega_1, \lambda}(z) p_{ji}} dz
\end{align*} 
In the following, we use the bound $(a+b)^2 \leq 2a^2 + 2b^2$ and the fact that $\beta \leq \alpha$ and $k \geq 1$. We have
\begin{align*}
	& \int_{\mathcal{Z}_{ji}} \frac{\{p_{\omega_1, \lambda}(z) - \overline{p}_{\omega_1}(z)\}^2}{p_{\omega_1, \lambda}(z) p_{ij}} dz \lesssim \left(\frac{h}{k^{1/d}}\right)^{2\beta} \frac{c^2_{hm}}{p_{ij}} \int\sum_{a, y} \frac{\one\{x\in \mathcal{C}_{h/2k^{1/d}}(m_{ji})\}  }{p_{\omega_1, \lambda}(z)} dz \\
	& \int_{\mathcal{Z}_{ji+m}} \frac{\{p_{\omega_1, \lambda}(z) - \overline{p}_{\omega_1}(z)\}^2}{p_{\omega_1, \lambda}(z)p_{ij}} dz \lesssim \left(\frac{h}{k^{1/d}}\right)^{2\beta} \frac{c^2_{hm}}{p_{i+mj}} \int\sum_{a, y} \frac{\one\{x\in \mathcal{C}_{h/2k^{1/d}}(m_{ji+m})\}}{p_{\omega_1, \lambda}(z)} dz 
\end{align*}
Let $\epsilon$ and $\overline\epsilon$ be such that 
\begin{align*}
	& \min\left\{ \left(\frac{1-|\theta|}{4} - \frac{1}{2}\left(\frac{h}{k^{1/d}}\right)^{\beta} - \frac{1 + |\theta|}{2}\left(\frac{h}{k^{1/d}}\right)^{\alpha} - \left(\frac{h}{k^{1/d}}\right)^{2s} \right) \right. , \\ & \left. \quad\quad\quad \left(\frac{1-|\theta|}{4} - \frac{1}{2}\left(\frac{h}{k^{1/d}}\right)^{\beta} - \frac{ h^{\gamma}}{4} \right) \right\} \geq \epsilon \\
	& \overline\epsilon = \frac{1 + |\theta|}{4} + \max\left\{ \frac{h^\gamma}{4}, \left(\frac{h}{k^{1/d}}\right)^{2s} \right\}
\end{align*}
Then, $p_{\omega_1, \lambda}(z) \geq c_{hm}\epsilon$ and
\begin{align*}
	\delta_1 \leq \frac{1}{\epsilon}\left(\frac{h}{k^{1/d}}\right)^{2\beta} \quad \text{and} \quad \frac{\overline{p}_{\omega_1}(z) }{p_{\omega_1, \lambda}(z)} \leq 1 \lor \frac{\overline\epsilon}{\epsilon}.
\end{align*}
Next, suppose $h^{\gamma - 2s}= 4k^{-2s / d}$. We have
\begin{align*}
	\delta_2 & \equiv \max_{ij}\sup_{\lambda} \int_{\mathcal{Z}_{ji}} \frac{\{p_{\omega_1, \lambda}(z) - p_{\omega_0, \lambda}(z)\}^2}{p_{1\lambda}(z) p_{ji}} dz \\
	& = \max_j \sup_{\lambda} \int_{\mathcal{Z}_{j1}} \frac{\{p_{\omega_1, \lambda}(z) - p_{\omega_0, \lambda}(z)\}^2}{p_{\omega_1, \lambda}(z) p_{j1}} dz
\end{align*} 
by symmetry and
\begin{align*}
	\int_{\mathcal{Z}_{j1}} \frac{\{p_{\omega_1, \lambda}(z) - p_{\omega_0, \lambda}(z)\}^2}{p_{\omega_1,\lambda}(z) p_{j1}} dz \lesssim \frac{1}{\epsilon}\left(\frac{h}{k^{1/d}}\right)^{2\alpha}
\end{align*}
because for any $i$: 
\begin{align*}
	& B\left( \frac{x - m_{ji}}{h / 2k^{1/d}} \right) = B\left( \frac{x - x_{i}}{h} \right) = 1 \quad \text{ if } x \in \cup_j \mathcal{C}_{h/2k^{1/d}}(m_{ji})
\end{align*}
Finally, again if $h^{\gamma - 2s}= 4k^{-2s / d}$, 
	\begin{align*}
		& \overline{p}_{\omega_0}(z)  - \overline{p}_{\omega_1}(z) = \\
		& = c_{hm} \one\{x \in \mathcal{S}_{hk}(x_1)\} (2a - 1)(2y - 1) \left\{\left(\frac{h}{k^{1/d}}\right)^{2s}  \sum_{j = 1}^k B\left( \frac{x - m_{j1}}{h / 2k^{1/d}} \right)^2 -  \frac{h^\gamma}{4} B\left( \frac{x - x_{1}}{h} \right) \right\} \\
		& - c_{hm} \one\{x \in \mathcal{S}_{hk}(x_{1 + m})\} (2a - 1)(2y - 1) \left\{\left(\frac{h}{k^{1/d}}\right)^{2s} \sum_{j = 1}^k B\left( \frac{x - m_{j{1 + m}}}{h / 2k^{1/d}} \right)^2 -  \frac{h^\gamma}{4} B\left( \frac{x - x_{1+m}}{h} \right) \right\} \\
		& = 0
	\end{align*}
By Lemma \ref{lemma:robins}, we conclude that
\begin{align*}
	H^2(p^n_{\omega_1}, p^n_{\omega_0}) & \leq \frac{C}{\epsilon^2}\frac{n^2 h^{2d}}{k} \left\{\left(\frac{h}{k^{1/d}}\right)^{4s} + \left(\frac{h}{k^{1/d}}\right)^{4\alpha}\right\} \leq \frac{2C}{\epsilon^2} \cdot n^2 h^{2d + 4s} k^{-1 - 4s/d}
\end{align*}
because $\alpha \geq s$. This bound on the Hellinger distance actually holds for any $p^n_{\omega'}$ and $p^n_{\omega}$ such that $\rho(\omega', \omega) = 1$. Finally, recall that $\xi \gamma \leq d$ by assumption. Set 
\begin{align*}
	&m = c_mh^{-d + \xi \gamma}, \quad  k = c_k n^{\frac{d}{2s} \cdot \frac{\gamma - 2s}{\gamma} \cdot \frac{1}{1 + d/(4s) + d/(2\gamma)}}, \quad h = \left(4c_k^{-\frac{2s}{d}}\right)^{\frac{1}{\gamma - 2s}}n^{-\frac{1}{\gamma\{1 + d/(4s) + d/(2\gamma)\}}}
\end{align*}
Notice that this choice enforces $h^{\gamma - 2s}= 4k^{-2s / d}$. Therefore,
\begin{align*}
	n^2 h^{2d + 4s} k^{-1-4s/d} & = \left(4c_k^{-\frac{2s}{d}}\right)^{\frac{2d + 4s}{\gamma - 2s}} \cdot c_k^{-1 - 4s / d} \cdot n^2 \cdot n^{-\frac{2d + 4s}{\gamma\{1 + d/(4s) + d/(2\gamma)\}}} \cdot n^{- \frac{d}{2s} \cdot \frac{\gamma - 2s}{\gamma} \cdot \frac{1 + 4s/d }{1 + d/(4s) + d/(2\gamma)}} \\
	& = \left(4c_k^{-\frac{2s}{d}}\right)^{\frac{2d + 4s}{\gamma - 2s}} \cdot c_k^{-1 - 4s / d}
\end{align*}
We can choose $c_k$ large enough so that $k \geq 1$, $h \leq1$ and the leading constant in the equation above is less than $\epsilon^2 / (2C)$. This way, the bound on the Hellinger distance is less than or equal to 1. Finally, we verify that $c_{hm}$ is finite. We need $2m$ disjoint cubes with sides equal to $2h$, so $m$ needs to satisfy $m \leq (2h)^{-d} / 2$ or, equivalently, $c_mh^{\gamma\xi} \leq 2^{-d-1}$. Because we can choose $h \leq 1$, choosing $c_m = 2^{-d - 2}(2^d - 2^{-d})^{-1}$ satisfies this requirement and yields
\begin{align*}
	& c_{hm} = \{1 - 2(2^d - 2^{-d}) h^d m)\}^{-1} \leq \{1 - 2(2^d - 2^{-d}) c_m)\}^{-1} = \left\{1 - \frac{1}{2^{d + 1}}\right\}^{-1} \\
	& \implies 1 \leq c_{hm} \leq \frac{4}{3}
\end{align*}
\end{document}